%% file: draft_arxiv.tex
\documentclass[12pt, hidelinks]{article}
\input{commands}
\usepackage{authblk}
\newcommand{\affilfont}{\small\itshape}

\title{Maximum-Variance-Reduction Stratification for Improved Subsampling}

\author[1,2,3]{Dingyi Wang}
\author[3]{Haiying Wang}
\author[1,2]{Qingpei Hu}

\affil[1]{\affilfont State Key Laboratory of Mathematical Sciences, Academy of Mathematics and Systems Science, Chinese Academy of Sciences, Beijing, China}
\affil[2]{\affilfont School of Mathematical Sciences,  University of Chinese Academy of Sciences, Beijing, China}
\affil[3]{\affilfont Department of Statistics, University of Connecticut, Storrs, CT, USA}

\date{}

\begin{document}
\maketitle

\begin{abstract}
  Subsampling is a widely used and effective approach for addressing the
  computational challenges posed by massive datasets. Substantial progress has
  been made in developing non-uniform, probability-based subsampling schemes
  that prioritize more informative observations. We propose a novel
  stratification mechanism that can be combined with existing subsampling
  designs to further improve estimation efficiency. We establish the estimator's
  asymptotic normality and quantify the resulting efficiency gains, which
  enables a principled procedure for selecting stratification variables and
  interval boundaries that target reductions in asymptotic variance. The
  resulting algorithm, Maximum-Variance-Reduction Stratification (MVRS),
  achieves significant improvements in estimation efficiency while incurring
  only linear additional computational cost. MVRS is applicable to both
  non-uniform and uniform subsampling methods. Experiments on simulated and real
  datasets confirm that MVRS markedly reduces estimator variance and improves
  accuracy compared with existing subsampling methods.
\end{abstract}
\keywords{Massive data, M-estimation, Optimal subsampling}

\section{Introduction}
The rapid advancement of information technology has catalyzed an exponential
surge in data volume. While this trend presents unprecedented opportunities for
scientific discovery, it simultaneously imposes significant challenges. When
applied to massive datasets, traditional statistical methods often incur
prohibitive computational costs and memory bottlenecks.

Subsampling has emerged as a powerful and widely adopted strategy to address
these constraints. By extracting an informative subset of the original data,
subsampling can significantly reduce computational costs while maintaining
statistical efficiency. Various subsampling approaches have been developed and
effectively implemented across a wide spectrum of models, such as linear
regression \citep{ma2015statistical, ma2022asymptotic, wang2019information},
generalized linear models \citep{fithian2014local, ai2021optimal1}, quantile
regression \citep{ai2021optimal2, wang2021optimal}, and semiparametric
regression \citep{breslow2007weighted, yu2022optimal, keret2023analyzing}.
Beyond parameter estimation, subsampling techniques are also used in
applications such as privacy protection \citep{balle2018privacy} and feature
selection \citep{zhu2022feature}. Readers can refer to \cite{YaoWang2021JDS,
yu2024review} for comprehensive reviews of subsampling methods.

To improve statistical efficiency, existing subsampling methods often utilize
non-uniform probabilities or specific design criteria to prioritize informative
observations. For instance, \cite{fithian2014local} proposed local case-control
sampling to target data points that are difficult to classify, while
\cite{ma2015statistical, ma2022asymptotic} focused on selecting observations
with high statistical leverage scores.  Explicitly pursuing estimation
optimality, \cite{wang2018optimal} introduced the optimal subsampling method
under A-optimality criterion (OSMAC), which minimizes the asymptotic mean
squared error of the subsample estimator.  Additionally, strategies leveraging
the covariate structure, such as prioritizing extreme covariate values
\citep{wang2019information} and utilizing orthogonal arrays
\citep{wang2021orthogonal}, have been employed to identify informative points.
Unlike these existing methods, which primarily focus on prioritizing individual
data points within the design, we explore the potential of informative
stratification to enhance estimation efficiency.

This paper proposes a novel stratification strategy that can be integrated with
existing subsampling designs to enhance estimation efficiency. We establish the
asymptotic properties of the resulting estimator under general M-estimation
settings, proving that efficiency improvement is guaranteed regardless of the
chosen stratification variable. To maximize these gains, we investigate the
optimal selection of stratification variables and intervals, proposing a
maximum-variance-reduction stratification (MVRS) scheme. Numerical studies on
both simulated and real-world datasets demonstrate the superiority of this approach.

The remainder of this paper is organized as follows. In
Section~\ref{sec:method}, we introduce the proposed stratified subsampling
framework and establish the asymptotic properties of the resulting estimator.
Section~\ref{sec.alg} discusses the selection of stratification variables and
intervals, and presents the practical MVRS algorithm. Sections~\ref{sim}
and~\ref{sec.cas} present simulations and real-world data analyses, respectively,
demonstrating the performance of the proposed method. Concluding remarks and
further discussions are provided in Section~\ref{sec.con}. All technical proofs
and additional numerical results are provided in the Appendix.

\section{Theoretical Framework and the Benefits of Stratification}
\label{sec:method}
\subsection{Non-uniform Subsampling Framework}
Assume the full dataset $\DN=\{X_i\}_{i=1}^N$ consists of $N$ independent
observations generated from $X\sim P_{\theta}$, where $P_{\theta}$ belongs to
the parametric family $\{P_{\theta},\theta \in \Theta\ \subset \mathbb{R}^d\}$.
Denote the empirical measure by $\Pr_N=N^{-1}\sumN\delta_{X_i}$, where
$\delta_{X_i}$ is a Dirac measure, i.e., $\delta_{X_i}(A)=\mathbb{I}_A(X_i)$ for
any measurable set $A$, with $\mathbb{I}$ being the indicator function. To
estimate $\theta$, we compute the M-estimator
\begin{equation}\label{est}
  \htheta_N=\arg\min_{\theta}\bigg[\Exp_{\Pr_N}\{l(X; \theta)\}
  =\oneN\sumN l(X_i; \theta)\bigg], 
\end{equation}
where $l$ is a loss function and the notation $\Exp_{\Pr_N}$ denotes the
expectation with respect to the measure $\Pr_N$. Since a closed-form solution
for $\htheta_N$ is generally not available, the Newton–Raphson method or other
iterative algorithms are usually employed to obtain a numerical solution. If the
data size $N$ is large, the computational cost can be prohibitive.

Subsampling methods aim to reduce the computational burden by selecting a
subsample of size $n$ (usually $\ll N$) from $\DN$ to replace the full dataset
used in \eqref{est}. Because observed data points in $\DN$ contain different
amounts of information about the parameter $\theta$, non-uniform subsampling
methods take more informative subsamples by assigning higher selection
probabilities to more informative observations. If we obtain a subsample with
replacement from $\DN$ according to probabilities $\{\pi_i\}_{i=1}^N$ with
$\sumN\pi_i=1$, then given $\DN$, the sampled observations $X_1^*,\dots,X_n^*$
are independent and identically distributed (i.i.d.) discrete random variables
that take the value $X_i$ with probability $\pi_i$ for $i=1,\dots,N$. Sampling
with replacement is used here because non-uniform subsampling without replacement
is computationally expensive for large $N$. Denote the subsample data as
$\D_n=\{X_i^*, \pi_i^*\}_{i=1}^n$, where $\pi_i^*$ is the realized subsampling
probability associated with $X_i^*$. The subsampling estimator is computed by
minimizing the inverse-probability-weighted target function. Denote the weighted
empirical measure by $\Q_N=\sumN \pi_i \delta_{X_i}$. The subsampling estimator
can be written as
\begin{equation*}\label{sub.est}
  \htheta_n^{\sub}=\arg\min_{\theta}
  \left[\Exp_{\Pr_N}\left\{\frac{\ud\Pr_N}{\ud\Q_N}\frac{N\mathbb{I}_{\D_n}(X)}{n}
    l(X; \theta)\right\}
  =\onen\sumn\frac{1}{N\pi_{j,i}^*}l(X_{j,i}^*; \theta)\right].
\end{equation*}
The estimator $\htheta_n^{\sub}$ is consistent to the full-data estimator under
some regularity conditions listed below. In the following, the notation
$^{\otimes2}$ denotes the outer product (i.e.,
$\dot{l}^{\otimes2}=\dot{l}\dot{l}\tp$), and $\dot{l}$ and $\ddot{l}$ represent
the gradient vector and Hessian matrix of $l$ with respect to $\theta$,
respectively.

\begin{assumption}\label{asmp01} The parameter space $\Theta$ is compact.
\end{assumption}

\begin{assumption}\label{asmp02} The risk (population loss) function
  $\Exp\{l(X;\theta)\}$ has a unique minimum, and $\Exp\{l^2(X;\theta)\}<\infty$
  for any $\theta \in \Theta$.
\end{assumption}

\begin{assumption}\label{asmp03} The matrix
  $\Exp\big\{\dot{l}^{\otimes2}(X;\theta)\big\}<\infty$ is positive-definite,
  and there exists $\delta>0$ such that
  $\Exp_{\Pr_N}\big\{\|\dot{l}(X;\htheta_N)\|^{2+\delta}\big\}=\Op$.
\end{assumption}

\begin{assumption}\label{asmp04} The matrix $\Exp\big\{\ddot{l}(X;\theta)\big\}$
  is positive-definite. The $(p,q)$-th element of $\ddot{l}(X;\theta)$ satisfies
  $\Exp\big\{\ddot{l}^2_{pq}(X;\theta)\big\}<\infty$, and there exists a $c(x)$
  satisfying $\Exp\{c^2(X)\}<\infty$ such that
  $|\ddot{l}^2_{pq}(x;\theta_1)-\ddot{l}^2_{pq}(x;\theta_2)| <
  c(x)\|\theta_1-\theta_2\|$ for any $\theta_1,\theta_2 \in \Theta$, and
  $p,q=1,2,...,d$.
\end{assumption}

\begin{assumption}\label{asmp05} The sampling probabilities satisfy $\max_{1 \le
  i \le N}(N\pi_i)^{-1}=\Op$.
\end{assumption}

Assumptions~\ref{asmp01}-\ref{asmp04} are commonly used regularity conditions.
Assumptions~\ref{asmp02}, \ref{asmp03}, and \ref{asmp04} essentially impose
moment conditions on the loss function and its first and second order
derivatives. Assumption~\ref{asmp05} imposes a constraint on the sampling
probabilities; it ensures that every data point has a relatively non-negligible
chance to be selected. Note that we do not attempt to impose the weakest
possible conditions here.

For the asymptotic properties of $\htheta_n^{\sub}$, we present a result from
\cite{wang2022sampling} as Proposition~\ref{sub.lem} to facilitate the
discussion.
 
\begin{proposition}\label{sub.lem} Under Assumptions~\ref{asmp01}-\ref{asmp05}
  in Section~\ref{sec:weight-strat-subs}, as $N \rightarrow \infty$ and $n
  \rightarrow \infty$, the estimator $\htheta_{n}^{\sub}$ satisfies
  \begin{equation*}
    \sqrt{n}(\V_{N}^{\sub})^{-1/2}(\htheta_{n}^{\sub}-\htheta_N)
    \rightarrow \Nor(0,I)
  \end{equation*}
  in distribution, where
  \begin{align*}
    \V_{N}^{\sub}
    =\Var_{\Q_N}\left\{\frac{\ud\Pr_N}{\ud\Q_N}\varphi(X;\htheta_N)\right\}
    =\oneN\sumN\frac{1}{N\pi_i}\varphi(X_i;\htheta_N)^{\otimes2},
  \end{align*}
  and $\varphi(X;\theta)=H^{-1}(\theta)\dot{l}(X;\theta)$ with
  $H(\theta)=\Exp_{\Pr_N}\{\ddot{l}(X;\theta)\}$.
\end{proposition}

The optimal subsampling procedure provides optimal subsampling probabilities
that minimize the trace of the conditional covariance matrix $\V_{N}^{\sub}$ in
various settings \citep[e.g.,][]{wang2018optimal, wang2021optimal,
ai2021optimal1, wang2022sampling}. We will show that the optimal subsampling
methods can be further improved by stratification, and we propose a
maximum-variance-reduction stratification (MVRS) for this purpose. Remarkably,
the MVRS is not limited to optimal subsampling methods; it can be used for any
given subsampling probabilities to further improve the estimation efficiency,
and it requires only linear additional computational time to implement.

\subsection{Asymptotics for Stratified Subsampling Estimators}
\label{sec:weight-strat-subs}
In this section, we demonstrate that for a general non-uniform subsampling
method, applying stratification improves the estimation efficiency of the
original estimator. We first introduce the proposed framework using an arbitrary
stratification variable, then derive the asymptotic properties of the resulting
estimator to demonstrate its efficiency gains. The choice of the stratification
variable is discussed in Section~\ref{sec.alg}.

Let $S$ be a stratification variable used to partition the dataset. Here, $S$
may be a function of $X$ or correlated with it. Let $\{S_i\}_{i=1}^N$ denote the
realizations of $S$ for the full dataset $\DN$. We partition the domain of $S$
into $k$ disjoint intervals $\bigcup_{j=1}^k A_j$, and denote the corresponding
partition of the index set as $I=\{1,\dots,N\}=\bigcup_{j=1}^k I_{j}$, where
$I_{j}=\{i \mid S_i \in A_j\}$.

We draw a subsample of size $n_j=\lfloor n\Pi_j+0.5\rfloor$ from each stratum
$\{X_i\}_{i \in I_{j}}$ with replacement, utilizing the sampling distribution
$\{\pi_i/\Pi_j\}_{i \in I_{j}}$. Here, $\Pi_j=\sum_{i \in I_{j}}\pi_i$ is the
stratum weight, $\sum_{j=1}^{k}n_j=n$, and $\lfloor\cdot\rfloor$ denotes the
floor function, ensuring $n_j$ is the nearest integer to $n\Pi_j$. The resulting
subsample is the union of the observations from each stratum, denoted as
$\D_n^{\str}=\bigcup_{j=1}^k\{(X_{j,1}^*, \pi_{j,1}^*), \dots, (X_{j,n_j}^*,
\pi_{j,n_j}^*)\}$.

Define the stratified subsample empirical measure as
\begin{equation*}
  \Pr_N^{\str}=\sumN\sum_{j=1}^k\frac{1}{n_j}\Pi_j\mathbb{I}_{A_j}(S_i)\delta_{X_i}.
\end{equation*}
The subsample estimator based on $\D_n^{\str}$ is given by
\begin{equation}\label{str.est}
\begin{split}
  \htheta_n^{\str}=\arg\min_{\theta}\Bigg[
  \Exp_{\Pr_N^{\str}}&\bigg\{\frac{\ud\Pr_N}{\ud\Q_N}
  \frac{N\mathbb{I}_{\D_n^{\str}}(X)}{n} l(X; \theta)\bigg\} \\
  &=\frac{1}{n}\sum_{j=1}^{k}\bigg\{\sum_{i=1}^{n_j}
  \frac{\Pi_j}{n_j\pi_{j,i}^*}l(X_{j,i}^*;\theta)\bigg\}\Bigg].
\end{split}
\end{equation}

The consistency and asymptotic normality of $\htheta_n^{\str}$ are established
in Theorem~\ref{str.thm}. The theorem indicates that the proposed estimator
$\htheta_{n}^{\str}$ is consistent with the full-data estimator $\htheta_N$ at a
$\sqrt{n}$-rate, with an asymptotic covariance matrix determined by the
stratification scheme.

\begin{theorem}\label{str.thm} Under Assumptions~\ref{asmp01}-\ref{asmp05}, as
  $N \rightarrow \infty$ and $n \rightarrow \infty$,
  $\htheta_{n}^{\str}-\htheta_N$ converges to zero in probability. Moreover,
  \begin{equation*}
    \sqrt{n}(\V_{N}^{\str})^{-1/2}(\htheta_{n}^{\str}-\htheta_N)
    \rightarrow \Nor(0, \I)
  \end{equation*}
  in distribution, where
  \begin{align*}
    \V_{N}^{\str}
    &=\Exp_{\Q_N}\left[
      \Var_{\Q_N}\left\{\frac{\ud\Pr_N}{\ud\Q_N}
      \varphi(X;\htheta_N)\bigg|S_A\right\}\right]\\
    &=\frac{1}{N^2}\sum_{j=1}^k\sum_{i \in I_j}\frac{\pi_i}{\Pi_j^2}
      \bigg\{\frac{\Pi_j}{\pi_i}\varphi(X_i;\htheta_N)
      -\sum_{i \in I_j}\varphi(X_i;\htheta_N)\bigg\}^{\otimes2},
  \end{align*}
  and $S_A$ is a discrete random variable defined as
  $S_{A_j}=j\mathbb{I}_{A_j}(S)$ for $j=1, \dots, k$ (i.e., $S_A$ indicates the
  stratum index of $S$).
\end{theorem}

We compare the estimation efficiency of $\htheta_n^{\sub}$ and
$\htheta_n^{\str}$ by examining their asymptotic covariance matrices in the
following theorem.

\begin{theorem}\label{str.thm2} If $\min_{j\in \{1, \dots, k\}}n\Pi_j\geq 1$,
  then the difference between the covariance matrices $\V_{N}^{\sub}$ (from
  Proposition~\ref{sub.lem}) and $\V_{N}^{\str}$ (from Theorem~\ref{str.thm})
  satisfies
  \begin{align}\label{diff}
    \V_{N}^{\str}-\V_{N}^{\sub}=-\Var_{\Q_N}\left[\frac{\ud\Pr_N}{\ud\Q_N}
      \Exp_{\Pr_N}\{\varphi(X;\htheta_N)\mid S_A\}\right]
      \leq 0,
  \end{align}
  where a matrix $B\leq 0$ means that $B$ is a negative semi-definite matrix.
\end{theorem}

Theorem~\ref{str.thm2} shows that for any sampling distribution
$\{\pi_i\}_{i=1}^N$, the proposed stratified subsampling estimator attains an
asymptotic covariance matrix that is no larger (in the positive semidefinite
sense) than that of the direct non-uniform subsampling method. While the two
covariance matrices may coincide in certain cases, Section~\ref{sec.Y}
demonstrates that a suitable choice of the stratification variable $S$ yields a
substantial reduction in variance.

\section{Maximum-Variance-Reduction Stratification}
\label{sec.alg}
We have established that the efficiency gains from stratification hinge on the
choice of the stratification variable $S$ and the defining intervals
$\{A_j\}_{j=1}^k$. In this section, we discuss the selection of these
components. Because the primary goal of subsampling is to reduce the
computational burden, we must avoid introducing excessive overhead during the
stratification phase. Therefore, the selection of $S$ and $\{A_j\}_{j=1}^k$
necessitates a trade-off between statistical efficiency and computational cost.
We present a practical algorithm implementing the proposed stratification
strategy and develop an estimator for the asymptotic covariance matrix
$\V_{N}^{\str}$.

\subsection{Stratification Variable}\label{sec.Y} The variance difference given
in \eqref{diff} of Theorem~\ref{str.thm2} implies that the efficiency gain from
stratification is more significant when the conditional expectation
$\Exp_{\Pr_N}\{\varphi(X;\htheta_N)\mid S_A\}$ exhibits larger
variability across strata. Ideally, one would choose $\varphi(X;\htheta_N)$ as
the stratification variable and group observations with similar
$\varphi(X;\htheta_N)$ values. While this is straightforward when $d=1$, there
is no natural ordering for $\varphi(X;\htheta_N)$ when $d>1$, and clustering
algorithms (e.g., K-means) become computationally prohibitive for large
datasets. Since the primary goal of subsampling is to reduce computational cost,
we recommend using a one-dimensional variable $S$ that captures the maximum
variation of $\varphi(X;\htheta_N)$. This approach improves estimation
efficiency without incurring substantial additional computational overhead.

We define $S$ as the linear transformation of $\varphi(X;\htheta_N)$ that
possesses the maximal variance under $\Pr_N$. Specifically, let
$S^{\mvrs}=\u\tp\varphi(X;\htheta_N)$, where
$\u=\arg\max_{\|\tilde\u\|=1}\Var_{\Pr_N}\{\tilde\u\tp\varphi(X;\htheta_N)\}$.
Because this strategy exploits the direction of the influence function with the
largest variance, we term our method Maximum-Variance-Reduction Stratification
(MVRS). Note that
$\max_{\tilde\u}\Var_{\Pr_N}\{\tilde\u\tp\varphi(X;\htheta_N)\}$ is the largest
eigenvalue of the covariance matrix $\Var_{\Pr_N}\{\varphi(X;\htheta_N)\}$, and
$\u$ is the associated eigenvector. Consequently, the stratification variable
$S$ can be obtained via a spectral decomposition of
$\Var_{\Pr_N}\{\varphi(X;\htheta_N)\}$.

\subsection{Stratification Interval}
\label{sec:choice-strat-interv}

The choice of stratification intervals is another critical component of the
proposed stratification scheme. We begin by examining the impact of the number
of strata on the resulting estimator. Let $\{A'_j\}_{j=1}^{k'}$ (where $k'>k$)
denote a refinement of the partition $\{A_j\}_{j=1}^{k}$, and let
$\htheta_n^{\str'}$ represent the estimator derived from this finer
stratification. Analogous to \eqref{diff}, the asymptotic covariance matrix of
$\htheta_n^{\str'}$, denoted as $n^{-1}\V_{N}^{\str'}$, is smaller than that of
$\htheta_n^{\str}$, and it can be shown that
\begin{align}\label{diff.2}
  \V_{N}^{\str'}-\V_{N}^{\str}
  =-\Exp_{\Q_N}\left[\Var_{\Q_N}\left\{\frac{\ud\Pr_N}{\ud\Q_N}\Exp_{\Q_N}\left(\varphi(X;\htheta_N)\Big|S'_A,S_A\right)\Bigg|S_A\right\}\right]\leq 0,
\end{align}
where $S'_A=j\mathbb{I}_{A'_j}(S)$ is a discrete random variable. A detailed
derivation of this result is provided in Section~\ref{pf.diff.2} of the
supplementary material.

This result implies that refining a stratum further reduces the asymptotic
variance. Consequently, the optimal number of strata $k$ is equal to the
subsample size $n$, with a single sample drawn from each stratum. Given that the
subsample size in stratum $j$ is $n_j=\lfloor n\Pi_j+0.5\rfloor$, optimal
stratification requires that $\lfloor n\Pi_j+0.5\rfloor=1$ for all
$j=1,\dots,n$. Notably, this scheme transforms subsampling with replacement into
subsampling without replacement.

This optimal stratification can be achieved by sorting the observations based on
$\{S_i\}_{i=1}^N$, partitioning the sorted data into $n$ strata
$\{A_j\}_{j=1}^n$ such that $\Q_N(S\in A_j)=n^{-1}$, and subsequently selecting
one observation from each stratum according to the sampling distribution
$\{\pi_i/\Pi_j\}_{i\in I_j}$. This implementation requires sorting all
observations according to $\{S_i\}_{i=1}^N$ and determining partitions according
to $\{\pi_i\}_{i=1}^N$, resulting in a time complexity of $O(N\log N)$. While
log-linear complexity is acceptable for many practical applications, we
demonstrate that it can be reduced further without significantly sacrificing
estimation efficiency.

To lower the computational cost, we select a moderate value for $k$ such that $k
\ll n$ and relax the uniformity requirement $\Q_N(S\in A_j)=n^{-1}$. Instead, we
partition $\{S_i\}_{i=1}^N$ into $k$ strata, each containing an equal number of
observations. This stratification scheme does not require evaluating the
$\pi_i$'s and achieves a time complexity of $O(N\log k)$ via a
divide-and-conquer algorithm analogous to quickselect
algorithm~\citep{hoare1961algorithm}. By recursively locating the quantiles and
partitioning the intervals, given the recursion depth as $O(\log k)$, the
algorithm achieves a time complexity of $O(N \log k)$. Since the asymptotic
variance decreases as $k$ increases, selecting $k<n$ involves a trade-off
between computational and estimation efficiency. In Section~\ref{sim}, we show
that even a small $k$ can yield a substantial improvement in estimation
efficiency.

\subsection{Practical Algorithm}
Since the variable $S^{\mvrs}=\u\tp\varphi(X;\htheta_N)$ depends on the
full-data estimator $\htheta_N$, it cannot be implemented directly. To address
this, we adopt a two-step procedure. The first step involves computing a pilot
estimator, which is then used in the second step to implement the proposed
stratification procedure. This pilot step is standard in most existing
informative subsampling methods, such as local case-control
\citep{fithian2014local} and OSMAC \citep{wang2018optimal}, which rely on a
preliminary estimate to approximate the informative sampling distribution.
Consequently, when applying the proposed stratification to these methods, the
existing pilot estimator suffices, incurring no additional computational cost.
For completeness, we describe the full algorithm here. First, a small uniform
subsample of size $n_0$ is drawn from the full dataset to compute a pilot
estimator $\htheta_{n_0}$. This estimator is used to construct the
stratification and, if necessary, the optimal subsampling distribution. Next,
stratified subsampling is performed to select the final subsample and obtain the
estimator. Details are provided in Algorithm~\ref{str.alg}, in which we use the
superscript $\tilde{}$ to indicate the pilot subsample and quantities that are
affected by the pilot estimation.

{
\renewcommand{\baselinestretch}{1}
\begin{algorithm}[H]
\caption{Practical MVRS Algorithm}
\label{str.alg}
\textbf{Step 1: Pilot construction and stratification}
\begin{enumerate}[{1}a):]
\item Take a uniform subsample $\{\tilde{X}_{i}\}_{i=1}^{n_0}$ to obtain a pilot
  estimate $\htheta_{n_0}$:
  \begin{equation}\label{eq:1}
    \htheta_{n_0}=\arg\min_{\theta}
    \frac{1}{n_0}\sum_{i=1}^{n_0}l(\tilde{X}_{i}; \theta).
  \end{equation}
\item [1a'):] (optional) Calculate subsampling probabilities $\tilde\pi_i$.
\item 
  \begin{enumerate}[(i)]
  \item 
      For $i=1,\dots,n_0$, calculate $\varphi(\tilde{X}_{i};\htheta_{n_0}) =
  \tilde{H}_0^{-1}\dot{l}(\tilde{X}_{i};\htheta_{n_0})$, where
  $\tilde{H}_0=n_0^{-1}\sum_{i=1}^{n_0} \ddot{l}(\tilde{X}_{i}; \htheta_{n_0})$.
\item 
  Obtain the eigenvector $\tilde\u$ for the largest eigenvalue of
  \begin{equation*}
    \frac{1}{n_0} \sum_{i=1}^{n_0}
    \varphi^{\otimes 2}(\tilde{X}_{i}; \htheta_{n_0}).
  \end{equation*}
\item 
  For $i=1,2,\dots,N$, calculate $\tilde S^{\mvrs}_i = \tilde\u\tp\varphi(X_i;
  \htheta_{n_0})$.
\item 
  For $j=1,\dots,k$, determine $\tilde I^{\mvrs}_j = \{i \mid\tilde S^{\mvrs}_{(j-1)}<\tilde S^{\mvrs}_i \le\tilde S^{\mvrs}_{(j)}\}$ with $\tilde S^{\mvrs}_{(j)}$ being the $j/k$ sample quantile of $\{\tilde S^{\mvrs}_i\}_{i=1}^N$ and $\tilde S^{\mvrs}_{(0)}=-\infty$, and calculate
  $\tilde\Pi_j = \sum_{i \in \tilde I^{\mvrs}_{j}}\tilde\pi_i$ and $\tilde n_j =
  \lfloor n \tilde\Pi_j + 0.5 \rfloor$.
  \end{enumerate}
\end{enumerate}
\textbf{Step 2: Subsampling and estimation}
\begin{enumerate}[2a):]
\item For $j = 1, \dots, k$, take $\tilde n_j$ observations with replacement
  from $\{X_i\}_{i \in \tilde I^{\mvrs}_{j}}$ according to subsampling
  distribution $\{\tilde\pi_i/\tilde\Pi_j\}_{i \in \tilde I^{\mvrs}_{j}}$,
  denoted as $X_{j,1}^*, X_{j,2}^*,\dots, X_{j,n_j}^*$.
\item Calculate
  \begin{equation}\label{eq:2}
    \htheta_n^{\mvrs} = \arg \max_{\theta} \sum_{j=1}^{k}
    \frac{1}{\tilde n_j} \sum_{i=1}^{\tilde n_j} \frac{\tilde\Pi_j}{\tilde\pi_{j,i}^*} l(X_{j,i}^*;
    \theta).
  \end{equation}
\end{enumerate}
\end{algorithm}
}

Step 1a') in Algorithm~\ref{str.alg} entails calculating the subsampling
probabilities, $\pi_i$, when they are not provided. This step is optional and
may be omitted if the $\pi_i$ are pre-specified, such as in uniform subsampling
where $\pi_i=N^{-1}$. Algorithm~\ref{str.alg} accommodates a wide range of
probability schemes, including leverage subsampling \citep{ma2015statistical}
and optimal subsampling \citep[e.g.,][]{wang2018optimal, wang2022sampling}.
Notably, generating optimal subsampling probabilities typically requires a pilot
estimator, $\htheta_{n_0}$, and the corresponding influence functions,
$\varphi(X_i;\htheta_{n_0})$. Consequently, the only additional computational
overhead of Algorithm~\ref{str.alg}—relative to standard optimal subsampling—is
incurred in Step 1b), which requires calculating the eigenvector $\u$, the
stratification variables $S_i$, and their sample quantiles.

We now analyze the computational complexity of the proposed MVRS. We assume that
$\log k < d$, which will be justified by our numerical result that a small $k$
is sufficient. Furthermore, we assume $nd < N$, reflecting the typical
subsampling regime where the sampling ratio $n/N$ approaches zero. Suppose that
parameter optimization employs a second-order iterative method (e.g., Newton's
method) with $\zeta$ iterations. The pilot estimation in Step 1a) solves
\eqref{eq:1} with complexity $O(\zeta_0 n_0 d^2)$; since $n_0 \ll N$, this cost
is negligible compared to $O(N)$. In Step 1b), calculating the stratification
variables $S_i$ requires linear time $O(Nd)$, and the stratification and
subsampling processes require $O(N \log k)$ time. The final estimation in Step 2
of solving \eqref{eq:2} has time complexity $O(\zeta n d^2)$. The total time
complexity depends on the cost of the optional Step 1a').  If Step 1a') requires
$O(Nd^2)$, the overall complexity of Algorithm~\ref{str.alg} is $O(\zeta_0n_0d^2
+ Nd^2 + N\log k + \zeta nd^2) = O(Nd^2)$; if Step 1a') requires linear time
$O(Nd)$, the total complexity is $O(Nd)$.  Notably, given any set of subsampling
probabilities, the MVRS adds only a linear overhead of $O(Nd)$, incurring no
significant additional computational cost.

\subsection{Theoretical Properties of the Practical
Estimator} \label{sec.mseest}

In Algorithm~\ref{str.alg}, the stratification variables $\tilde S^{\mvrs}_i$
are constructed using $\htheta_{n_0}$, and the subsampling probabilities
$\tilde\pi_i$ may rely on this pilot estimator as well such as in optimal
subsampling. Since this dependence on $\htheta_{n_0}$ can influence the
asymptotic behavior of the final result, we explicitly establish the theoretical
properties of the estimator $\hat\theta_n^{\mvrs}$ in Theorem~\ref{alg.thm}.

\begin{theorem}\label{alg.thm} Under Assumptions~\ref{asmp01}-\ref{asmp05}, if
  $\Exp\{\varphi^4(X;\theta)\}\leq \infty$ and the distribution function of
  $\u\tp\varphi(X;\theta)$ is continuous with positive density at its
  $j/k$-quantiles for $j=1, \dots, k-1$, then as $N, n, n_0 \rightarrow \infty$,
  \begin{equation*}
    \sqrt{n}(\V_{N}^{\mvrs})^{-1/2}(\hat\theta_n^{\mvrs}-\htheta_N)
    \rightarrow \Nor(0, \I)
  \end{equation*}
  in distribution, where
  \begin{align*}
    \V_{N}^{\mvrs}
    &=\Exp_{\Q_N}\left[
    \Var_{\Q_N}\bigg\{\frac{\ud\Pr_N}{\ud\Q_N}
      \varphi(X;\htheta_N)\bigg|S^{\mvrs}_{A^{\mvrs}}\bigg\}\right]\\
    &=\oneN\sum_{j=1}^k\sum_{i \in I^{\mvrs}_j}
      \frac{\pi_i}{N\Pi_j^2}
    \bigg\{\frac{\Pi_j}{\pi_i}\varphi(X_i;\htheta_N)
      -\sum_{i \in I^{\mvrs}_j}\varphi(X_i;\htheta_N)\bigg\}^{\otimes2}.
  \end{align*}
  Here, $S^{\mvrs}_{A^{\mvrs}}$ is defined as
  $S^{\mvrs}_{A^{\mvrs}}=j\mathbb{I}_{A^{\mvrs}_j}(S^{\mvrs})$ , where
  $A^{\mvrs}_j = (S^{\mvrs}_{(j-1)}, S^{\mvrs}_{(j)}]$ for $j=1, ..., k$, and
  $I^{\mvrs}_j=\{i \mid S^{\mvrs}_{(j-1)}<S^{\mvrs}_i \le S^{\mvrs}_{(j)}\}$.
\end{theorem}

If the number of strata is $k=1$, then $\htheta_n^{\mvrs}$ reduces to the
two-step non-uniform subsampling estimator without stratification, where the
subsampling probabilities $\tilde{\pi}_i$ are estimated using the pilot
estimator $\hat{\theta}_{n_0}$ rather than the full-data estimator $\htheta_N$.
Consequently, Theorem~\ref{alg.thm} reduces to a version of
Proposition~\ref{sub.lem} for the practical subsampling estimator.  It follows
that the practical MVRS estimator achieves an asymptotic variance no larger than
that of the corresponding practical subsampling estimator without
stratification. This result is analogous to Theorem~\ref{str.thm2} and is
verified by the inequality:
\begin{align*}
      \V_{N}^{\mvrs} -\V_{N}^{\sub}
  &=-\Var_{\Q_N}\left[\frac{\ud\Pr_N}{\ud\Q_N}\Exp_{\Pr_N}\left\{\varphi(X;\htheta_N)\Big|S^{\mvrs}_{A^{\mvrs}}\right\}\right]
  \leq 0.
\end{align*}

To quantify the uncertainty of $\hat\theta_n^{\mvrs}$ or to facilitate inference
(e.g., constructing confidence intervals), we must estimate its asymptotic
variance. The quantity $\V_{N}^{\mvrs}$ in Theorem~\ref{str.thm} cannot be
computed directly as it depends on the full-data estimator $\htheta_N$. We
provide a feasible estimator, $\hat{\V}_{N}^{\mvrs}$, using only the selected
subsample:
\begin{equation}\label{eq:3}
  \hat{\V}_{N}^{\mvrs}=\hat{H}_N^{-1}\hat{\Phi}^{\mvrs}_N (\hat{H}_N^{-1})\tp,
  \text{ where }
    \hat{H}_N=\oneN\sum_{j=1}^{k}\frac{\tilde \Pi_j}{\tilde n_j}\sum_{i=1}^{\tilde n_j} \frac{1}{\tilde \pi_{j,i}^*}\ddot{l}(X_{j,i}^*; \htheta_n^{\mvrs}),
\end{equation}
and
\begin{equation}\label{eq:4}
  \hat{\Phi}^{\mvrs}_N
  =\frac{n}{n-d}\frac{1}{N^2}\sum_{j=1}^k\frac{1}{\tilde n_j}\sum_{i=1}^{\tilde n_j}\tilde \Pi_j\left\{
        \frac{1}{\tilde \pi_{j,i}^*}\dot{l}(X_{j,i}^*;\htheta_n^{\mvrs})
        -\frac{1}{\tilde n_j}\sum_{i=1}^{\tilde n_j}\frac{1}{\tilde \pi_{j,i}^*}\dot{l}(X_{j,i}^*;\htheta_n^{\mvrs})\right\}^{\otimes2}.
\end{equation}
In \eqref{eq:4}, the term $n/(n-d)$ is a finite-sample correction for the
degrees-of-freedom, which is useful when the subsample size $n$ is not
significantly large relative to the parameter dimension $d$.

\section{Simulation}\label{sim} We conduct simulations to evaluate the
performance of the proposed MVRS method, applying it to both uniform and optimal
subsampling probabilities. Section~\ref{sec.log} and~\ref{sec.svm} present the
performance of different methods for logistic regression and support vector
machine. The estimation efficiency is compared against the corresponding
subsampling methods without stratification by examining the mean squared error
(MSE) of the resulting estimators. The impact of the number of strata is
investigated in Section~\ref{sim2}. In Section~\ref{sec.time} we record the
computational time for different approaches. Finally, Section~\ref{sim3}
assesses the accuracy of the estimated asymptotic covariance matrix
$\hat{\V}_{N}^{\str}$.

\subsection{Logistic Regression}\label{sec.log}
We generate i.i.d. full data $\DN=\{X_i=(z_i,y_i)\}_{i=1}^N$ from logistic
regression, where $z_i$ denotes the covariate vector and $y_i$ denotes the
response. The loss function is the negative log-likelihood as
\begin{equation*}
  l(X_i; \theta)
  =-y_i(\theta_0+\theta_1\tp z_i)+\log(1+e^{\theta_0+\theta_1\tp z_i}),
\end{equation*}
where $\theta=(\theta_0,\theta_1\tp)\tp$ is the vector of regression
coefficients, with $\theta_0$ representing the intercept and $\theta_1$ the
slope parameters.

We set the full data sample size to $N=5\times 10^5$ and the true regression
coefficients $\theta=0.1\times\mathbf{1}_{15}$, where $\mathbf{1}_{15}$ means a
vector of ones with dimension 15. We use the four covariate
distributions that are used in \cite{wang2018optimal} to generate
$z_i=(z_{i1},...,z_{i14})$:
\begin{enumerate}[{Case }1:]
\item The covariates $z_i \stackrel{i.i.d.}{\sim} \Nor(\textbf{0},\Sigma)$,
  where $\Sigma_{ij}=0.5^{\mathbb{I}(i\neq j)}$ and $\mathbb{I}(\cdot)$ is the
  indicator function.
\item The covariates $z_i \stackrel{i.i.d.}{\sim} \Nor(\textbf{1.5},\Sigma)$,
  where $\Sigma$ is the same as in Case 1.
\item The covariates
  $z_i \stackrel{i.i.d.}{\sim} \Nor(\textbf{0},\Sigma_u)$, where
  $\Sigma_{ij}=0.5^{\mathbb{I}(i\neq j)}/(ij)$. Components of $z_i$ have unequal
  variances in this case.
\item The components of $z_i$ are $z_{ij} \stackrel{i.i.d.}{\sim}
  \mathbb{EXP}(2)$.
\end{enumerate}

We consider subsample sizes of $n=\{1000, 1500, 2000, 2500\}$ and fix the pilot
subsample size at $n_0=500$.  
 We set the number of strata to $k=10$;
Section~\ref{sim2} demonstrates that this choice is sufficient to achieve
desirable estimation efficiency.

We compare the proposed method against two representative subsampling methods:
uniform subsampling (UNIF) and optimal subsampling (OPT) \citep{wang2018optimal,
wang2022sampling}. For UNIF, the subsampling probabilities are $\pi_i=N^{-1}$.
For OPT, the subsampling probabilities satisfy $\pi_i \propto
\|\varphi(X_i;\htheta_N)\|$. When integrating the MVRS scheme with UNIF and OPT,
we denote the corresponding subsampling methods as MVRS-U and MVRS-O,
respectively. To evaluate the performance of each estimator, we repeat the
simulation $R=1000$ times and calculate the empirical mean squared error as
$\operatorname{MSE}=R^{-1}\sum_{r=1}^R\| \htheta_{n,r}-\htheta_N\|^2$, where
$\htheta_{n,r}$ is the estimator obtained from the $r$-th simulation and
$\htheta_N$ is the full data estimator.

Figure~\ref{fig:log} displays the results for logistic regression. In all
settings, MVRS-U and MVRS-O consistently outperform their corresponding baseline
methods (UNIF and OPT), confirming that MVRS effectively enhances estimation
efficiency for a given baseline subsampling distribution. Notably, MVRS-U surpasses OPT,
the most efficient unstratified method, in some cases, which highlights the
substantial efficiency gains provided by the proposed MVRS. In Case 3, the
efficiency gain from MVRS is particularly pronounced.  Unlike the baseline OPT
method, which yields no significant improvement over the baseline UNIF, MVRS-U
achieves a notable increase in estimation efficiency.  
Ultimately, MVRS-O achieves superior performance
across all scenarios, indicating that combining OPT with MVRS yields the most
efficient estimator.

\begin{figure}[H]
    \centering
    \begin{subfigure}{0.45\linewidth}
        \centering
        \includegraphics[width=\linewidth]{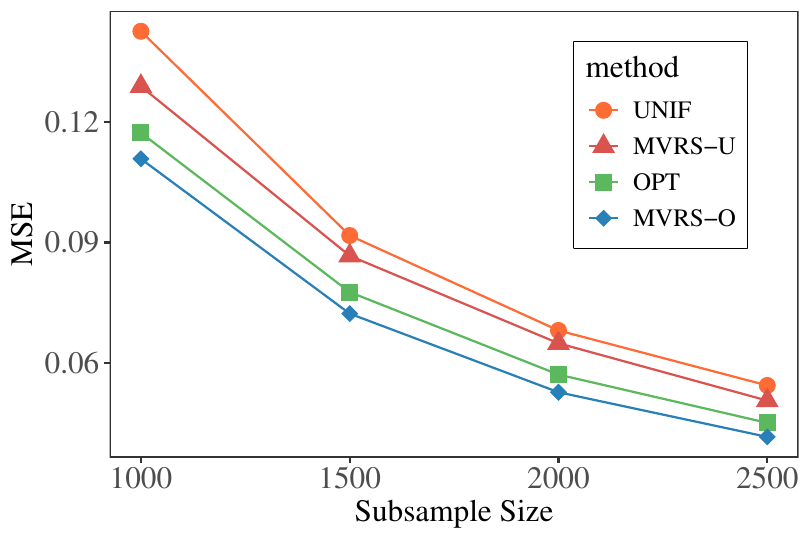}
        \caption{Case 1 (mzNormal)}
    \end{subfigure}
    \hspace{0.05\linewidth}
    \begin{subfigure}{0.45\linewidth}
        \centering
        \includegraphics[width=\linewidth]{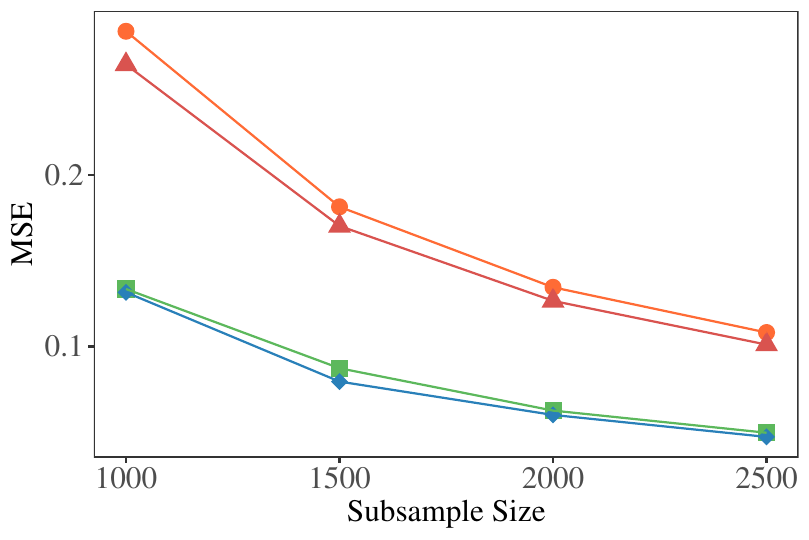}
        \caption{Case 2 (nzNormal)}
    \end{subfigure}
    \vfill
    \begin{subfigure}{0.45\linewidth}
        \centering
        \includegraphics[width=\linewidth]{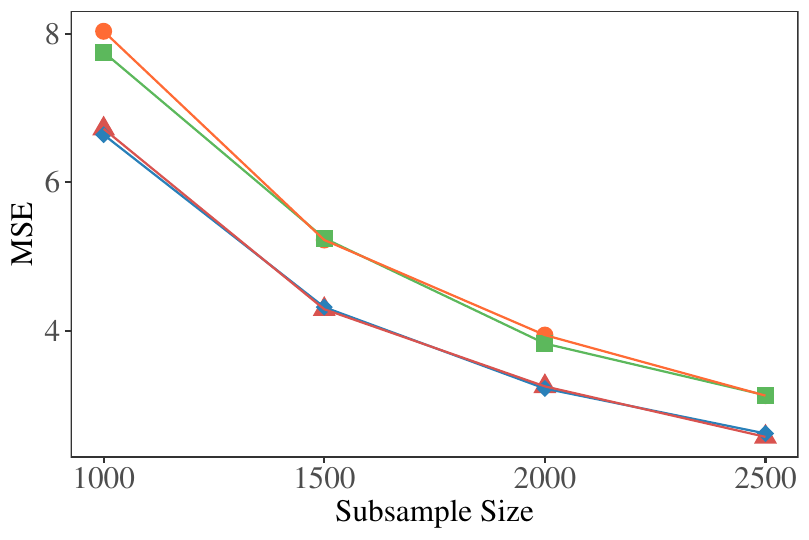}
        \caption{Case 3 (ueNormal)}
    \end{subfigure}
    \hspace{0.05\linewidth}
    \begin{subfigure}{0.45\linewidth}
        \centering
        \includegraphics[width=\linewidth]{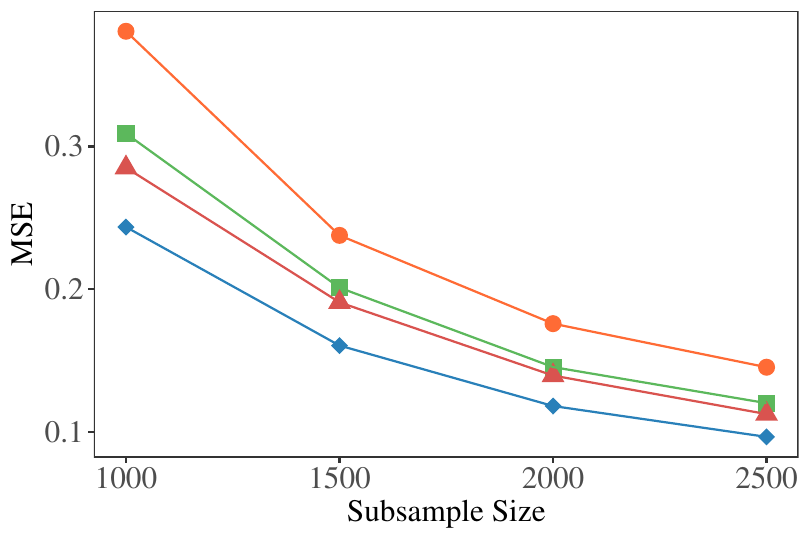}
        \caption{Case 4 (EXP)}
    \end{subfigure}
    \caption{MSEs for different subsample sizes $n$ in logistic regression.}
    \label{fig:log}
\end{figure}

To further evaluate the robustness of the proposed MVRS, we consider two
additional scenarios:
\begin{enumerate}[{Case }1:]
  \setcounter{enumi}{4}
\item Heterogeneous data. The covariates are generated from two different
  distributions: half of the covariates are generated from the
  normal distribution as in Case 1, and the other half are generated from a
  Binomial distribution $\mathbb{BIN}(1,0.5)$. 
\item Misspecified model. We generate the responses $y_i$ from a wrong model
  based on the covariate distribution in Case 5. Specifically, we adopt probit
  regression for the normally distributed covariates and logistic regression for
  the binomial covariates. The response generation follows a hybrid mechanism
  but we still use a logistic regression for estimation.
\end{enumerate}

\begin{figure}[H]
    \centering
    \begin{subfigure}{0.45\linewidth}
        \centering
        \includegraphics[width=\linewidth]{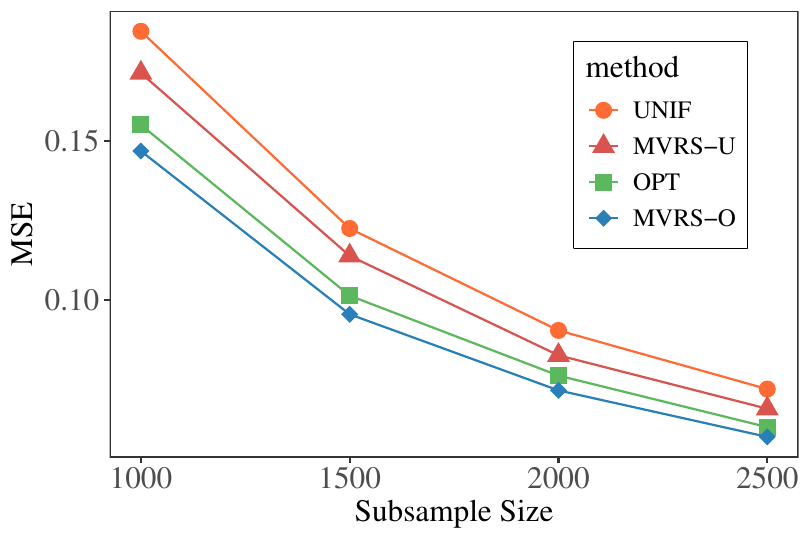}
        \caption{Case 5: Heterogeneous data}
    \end{subfigure}
    \hspace{0.05\linewidth}
    \begin{subfigure}{0.45\linewidth}
        \centering
        \includegraphics[width=\linewidth]{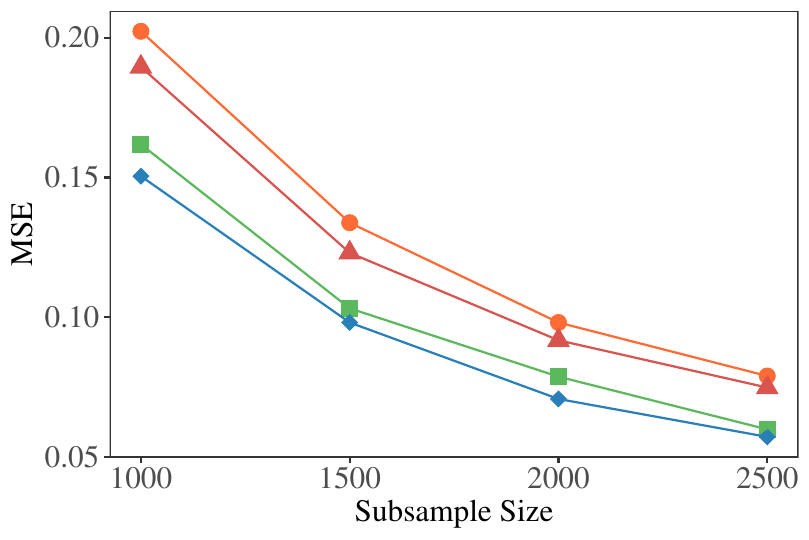}
        \caption{Case 6: Misspecified model}
    \end{subfigure}
    \caption{MSEs for heterogeneous data and misspecified model in logistic regression.}
    \label{fig:mis}
\end{figure}

Figure~\ref{fig:mis} presents the MSEs for Case 5 and Case 6.  The results
indicate that the proposed MVRS method maintains its effectiveness in boosting
the efficiency over baseline methods in both scenarios, demonstrating its
robustness to data heterogeneity and model misspecification.  Note that under
model misspecification, a true parameter no longer exists within the class of
logistic regression models. However, the full-data estimator $\htheta_N$ remains
a reasonable target because it converges to the least false limit $\theta_{KL}$,
which minimizes the Kullback--Leibler divergence between the logistic regression
model and the true data-generating process \citep{white1982maximum}. Our
real-world applications in Section~\ref{sec.cas} further confirm the robustness
and advantages of the proposed method in practical settings, where data
distributions are often more complex than simulated scenarios and the true
data-generating process is unknown.

We also conducted simulations evaluating information-based optimal subdata
selection (IBOSS) \citep{cheng2020information}.  Because the proposed MVRS
framework is specifically designed for random subsampling, these results are
provided in the supplementary material.

\subsection{Support Vector Machines}\label{sec.svm} 
To evaluate the performance of the proposed method in more complex scenarios, we
consider nonlinear support vector machines (SVMs). The full-data size is set to
$N = 5\times 10^5$ with $d=10$ baseline covariates, which are expanded to $p =
65$ predictors by including quadratic and interaction terms. The covariates are
generated from a normal distribution such that $z_{ij} \sim N((d+1-j)/2,
1/j^2)$, and the correlation between $z_{ij}$ and $z_{ik}$ is $0.5$ for $j \neq
k$.  The binary responses $y_i \in \{1, -1\}$ are generated according to four
different decision boundaries:
\begin{enumerate}[{Boundary }1:]
\item \begin{equation*}
0.1 \sum_{j=1}^{10} z_{ij} + 0.1 \sum_{k \in \{1,3,5,7,9\}} z_{i,k}z_{i,k+1} - C_1,
\end{equation*}
\item\begin{equation*}
0.1 \sum_{j=1}^{10} j \cdot z_{ij} + 0.1 \sum_{k \in \{1,3,5,7,9\}} z_{i,k}z_{i,k+1} - C_2,
\end{equation*}
\item\begin{equation*}
0.1 \sum_{j=1}^{10} z_{ij} + 0.1 \sum_{k \in \{1,3,5\}} \sin(z_{i,k}z_{i,k+1}) + 0.1 \sum_{l \in \{7,9\}} \exp(z_{i,l}z_{i,l+1}) - C_3,
\end{equation*}
\item\begin{equation*}
0.1 \sum_{j=1}^{10} z_{ij} + 0.1 \sum_{k \in \{1,3,5\}} \text{sign}(z_{i,k}z_{i,k+1}) - C_4,
\end{equation*}
\end{enumerate}
where $C_1, \dots, C_4$ are centering constants. We employ the squared hinge
loss as the objective function; all other experimental settings follow those
described for logistic regression. Figure~\ref{fig:svm} presents the results,
which show that the performance of the proposed MVRS is consistent with the
findings for logistic regression: MVRS-U and MVRS-O consistently outperform
their respective baseline methods (UNIF and OPT) across all decision boundaries.
The efficiency gains from MVRS are particularly pronounced for Boundary 3, which
features a more complex, nonlinear decision boundary. These results indicate
that MVRS can effectively enhance estimation efficiency in complex tasks as
well.

\begin{figure}[H]
    \centering
    \begin{subfigure}{0.45\linewidth}
        \centering
        \includegraphics[width=\linewidth]{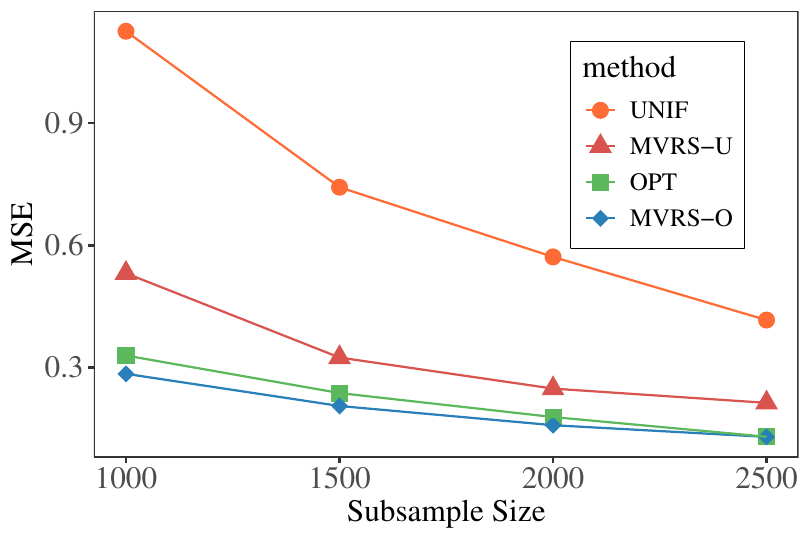}
        \caption{Boundary 1 (linear boundary)}
    \end{subfigure}
    \hspace{0.05\linewidth}
    \begin{subfigure}{0.45\linewidth}
        \centering
        \includegraphics[width=\linewidth]{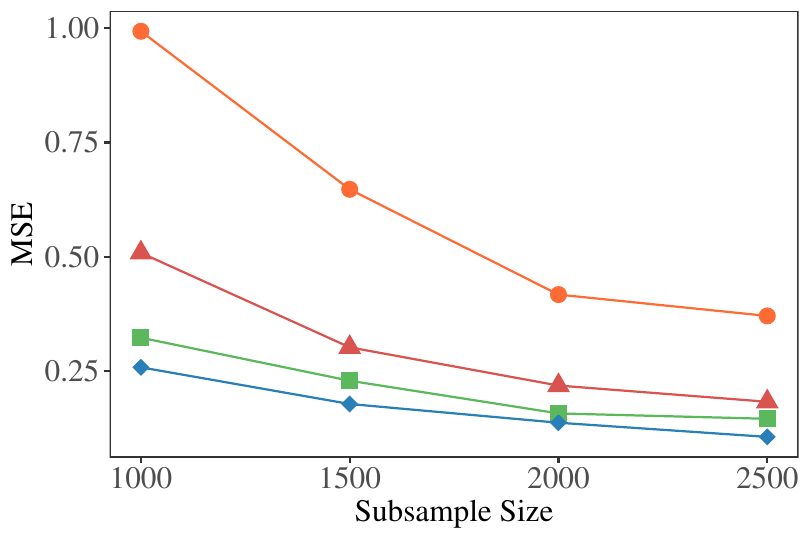}
        \caption{Boundary 2 (different weights)}
    \end{subfigure}
    \vfill
    \begin{subfigure}{0.45\linewidth}
        \centering
        \includegraphics[width=\linewidth]{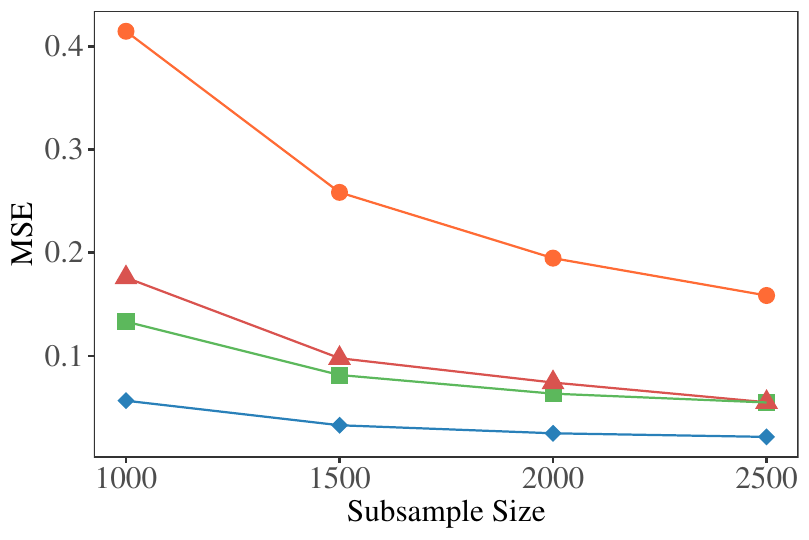}
        \caption{Boundary 3 (non-linear boundary)}
    \end{subfigure}
    \hspace{0.05\linewidth}
    \begin{subfigure}{0.45\linewidth}
        \centering
        \includegraphics[width=\linewidth]{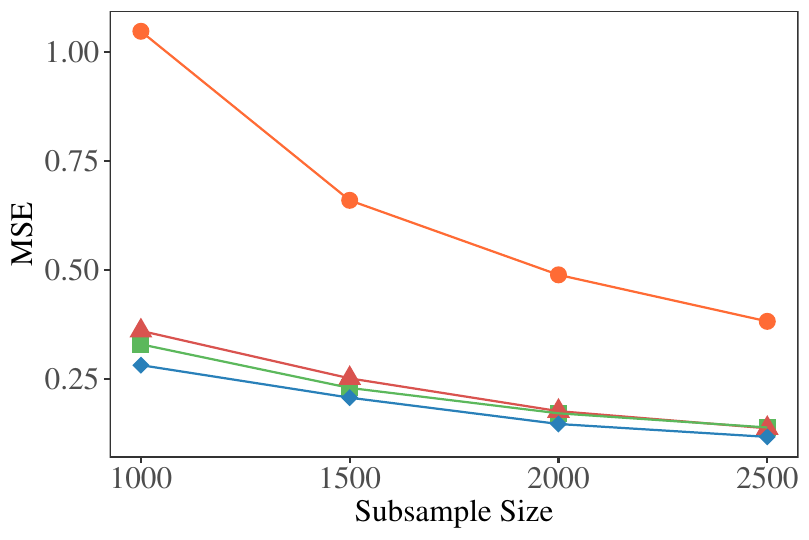}
        \caption{Boundary 4 (jump boundary)}
    \end{subfigure}
    \caption{MSEs for different subsample sizes $n$ in support vector machine.}
    \label{fig:svm}
\end{figure}

\subsection{Effects of the Number of Strata $k$}\label{sim2} We now investigate
the impact of the number of strata, $k$, on estimation efficiency.
Figure~\ref{fig:str} presents the results for logistic regression with stratum
counts $k \in \{1, 5, 10, 50, 100\}$ and subsample sizes of $n = 1000$. We see
that increasing the number of strata generally improves estimation accuracy,
which aligns with the theoretical derivation in
Section~\ref{sec:choice-strat-interv}. However, the marginal efficiency gain
diminishes as the number of strata increases beyond a certain threshold. The
estimation efficiency improves significantly as $k$ increases from $k=1$ (no
stratification) to $k=10$. While increasing $k$ to $k=50$ provides a slight
additional boost, the improvement is not substantial. Given that computational
complexity scales with the number of strata, a small $k$ offers a superior
tradeoff between estimation efficiency and computational cost. Consequently,
$k=10$ is a suitable choice for the considered simulation setup. Results for
other subsample sizes exhibited similar patterns and are omitted to save space.

\begin{figure}[H]
    \centering
    \begin{subfigure}{0.45\linewidth}
        \centering
        \includegraphics[width=\linewidth]{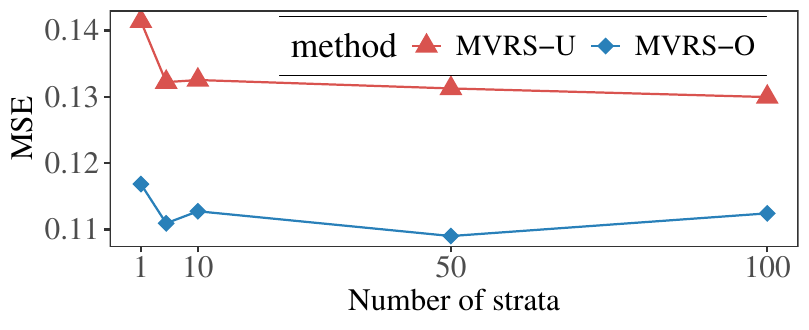}
        \caption{Case 1 (mzNormal)}
    \end{subfigure}
    \hspace{0.05\linewidth}
    \begin{subfigure}{0.45\linewidth}
        \centering
        \includegraphics[width=\linewidth]{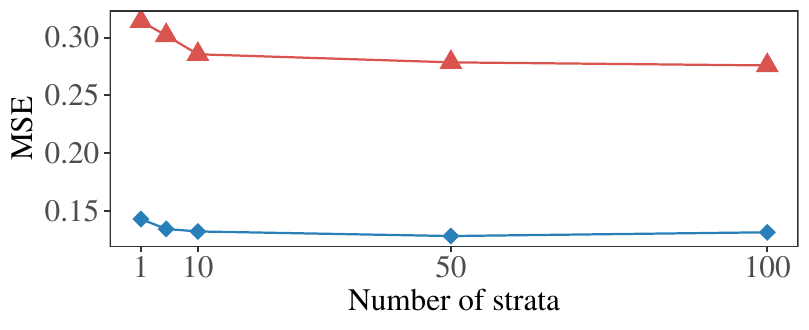}
        \caption{Case 2 (nzNormal)}
    \end{subfigure}
    \vfill
    \begin{subfigure}{0.45\linewidth}
        \centering
        \includegraphics[width=\linewidth]{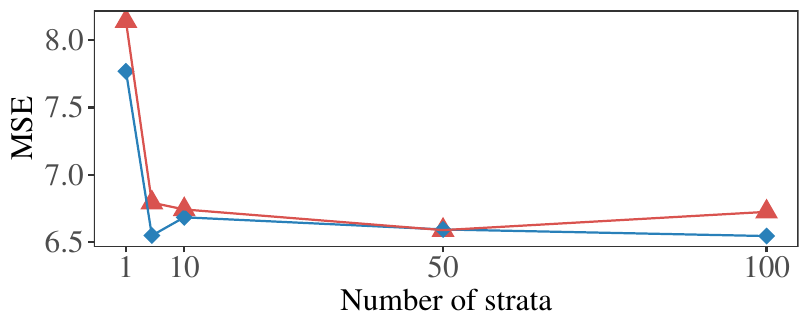}
        \caption{Case 3 (ueNormal)}
    \end{subfigure}
    \hspace{0.05\linewidth}
    \begin{subfigure}{0.45\linewidth}
        \centering
        \includegraphics[width=\linewidth]{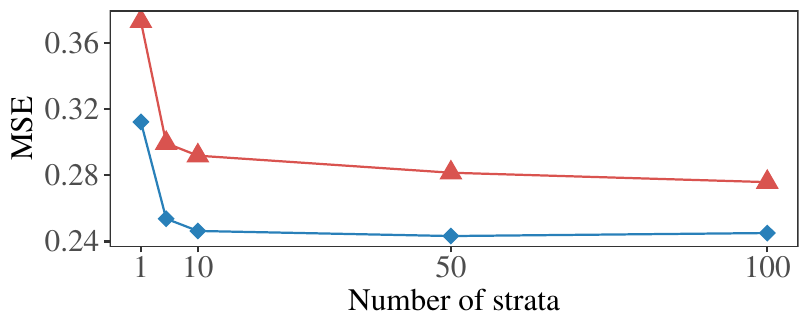}
        \caption{Case 4 (EXP)}
    \end{subfigure}
    \caption{MSEs for logistic regression with different numbers of strata $k$.}
    \label{fig:str}
\end{figure}

To further evaluate the performance of the proposed MVRS, we compare it with the
optimal stratification method discussed in
Section~\ref{sec:choice-strat-interv}. This benchmark method sets the number of
strata equal to the subsample size ($k=n$) and partitions the data using equal
probabilities (i.e., $\Q_N(S\in A_j)=1/n$) rather than equal numbers of
observations. Implementing this procedure requires ranking the entire dataset
and evaluating probability values, resulting in a higher computational cost.
Table~\ref{tab:mse} summarizes the MSEs for Case 1 of logistic regression, where
the optimal methods are denoted as OMVRS-U and OMVRS-O. The results demonstrate
that the MVRS attains performance comparable to the optimal stratification
method, indicating that the proposed stratification scheme is near-optimal while
requiring lower computational resources. Results for other cases are similar and
are omitted for brevity.

\begin{table}[H]
\centering
\caption{MSEs of the proposed methods and the optimal stratification estimator
  in Case 1.}
\label{tab:mse}
\resizebox{0.6\textwidth}{!}{
\begin{tabular}{c|c c c c}
    \hline
    \multirow{2}{*}{Method}& \multicolumn{4}{c}{MSE}\\
    & $n=1000$ & $n=1500$ & $n=2000$ & $n=2500$ \\
    \hline
    UNIF &0.142& 0.092& 0.068& 0.054 \\
    MVRS-U &0.129& 0.087& 0.065& 0.051 \\
    OMVRS-U &0.130 &0.086& 0.063& 0.050 \\
    OPT & 0.117& 0.078& 0.057& 0.045 \\
    MVRS-O &0.111& 0.072& 0.053& 0.042 \\
    OMVRS-O &0.110& 0.072& 0.054& 0.042 \\
    \hline
\end{tabular}
}
\end{table}

\subsection{Computation Time}\label{sec.time}
Tables~\ref{tab:time_log} and \ref{tab:time_svm} present the computation times
for the full data method (FULL) and the subsample methods ($n=2500$) across
varying numbers of strata. All simulations were implemented in the R programming
language \citep{R} on a laptop equipped with an Intel Core Ultra 9 185H
processor (2.30 GHz) and 32GB of RAM. We used the standard \texttt{order()}
function in R for the stratification step; while this is not the most efficient
implementation for the stratification step, the resulting runtimes remain
acceptable. As shown in Tables~\ref{tab:time_log} and \ref{tab:time_svm}, MVRS-U
and MVRS-O incur a slight computational overhead compared to the baseline UNIF
and OPT methods due to the stratification process, yet they remain significantly
faster than the full data analysis. 

For MVRS-U, no computational overhead is needed to calculate the values of
$\tilde\Pi_j$ because they are identical and known. Consequently, the
computation cost in Step 1 after stratification is negligible. As a result, the
computation time does not increase significantly with the number of strata.  For
MVRS-O, while computation time naturally increases with the number of strata,
the cost remains manageable even at $k=100$, demonstrating the practical
feasibility of the proposed method. The optimal stratification methods (OMVRS-U
and OMVRS-O) require more time than the corresponding MVRS methods due to
evaluating probabilities. However, their runtimes remain reasonable. Results for
other cases are omitted due to similarity.

Comparing the runtimes in Tables~\ref{tab:time_log} and \ref{tab:time_svm}
highlights that subsampling is particularly valuable for computationally
intensive models like SVM. For logistic regression, while subsampling is much
faster than full data analysis, the latter is already efficient enough that the
time savings may not be significant in practice. Conversely, full data analysis
for SVM is prohibitively slow, making the proposed subsampling methods essential
for feasible analysis. Furthermore, although the UNIF method remains faster than
other subsampling schemes, its relative advantage is less significant for SVM
than for logistic regression. This is because the estimation step (rather than
the sampling process) becomes the dominant computational cost in more complex
models.

\begin{table}[H]
 
\centering
\caption{  Computational times (in seconds) of different methods with $n=2500$ for
   Case 1 of logistic regression.}
\label{tab:time_log}
\resizebox{0.65\textwidth}{!}{
\begin{tabular}{c c c c c c c}
    \toprule
    Method & $k$ & Time (s) && Method & $k$ & Time (s) \\
    \midrule
    FULL & -- & 0.757 && & & \\ 
    \midrule
    UNIF & -- & 0.004 && OPT & -- & 0.232 \\
    \midrule
    \multirow{4}{*}{MVRS-U} & 5   & 0.099 && \multirow{4}{*}{MVRS-O} & 5   &
    0.276 \\
    & 10  & 0.096 &   && 10  & 0.275 \\
    & 50  & 0.100 &   && 50  & 0.276\\
    & 100 & 0.100 &   && 100 & 0.269 \\
    \midrule
    OMVRS-U & -- & 0.139 && OMVRS-O & -- & 0.309 \\
    \bottomrule
\end{tabular}
}
\end{table}

\begin{table}[H]
 
\centering
\caption{  Computational times (in seconds) of different methods with $n=2500$ for
   Boundary 1 of SVM.}
\label{tab:time_svm}
\resizebox{0.65\textwidth}{!}{
\begin{tabular}{c c c c c c c}
    \toprule
    Method & $k$ & Time (s) && Method & $k$ & Time (s) \\
    \midrule
    FULL & -- & 304.88 && & & \\ 
    \midrule
    UNIF & -- & 1.07 && OPT & -- & 3.63 \\
    \midrule
    \multirow{4}{*}{MVRS-U} & 5   & 3.59 && \multirow{4}{*}{MVRS-O} & 5   & 3.72
    \\
    & 10  & 3.59 &   && 10  & 3.76 \\
    & 50  & 3.58 &   && 50  & 3.72\\
    & 100 & 3.60 &   && 100 & 3.74 \\
    \midrule
    OMVRS-U & -- & 3.51 && OMVRS-O & -- & 3.68 \\
    \bottomrule
\end{tabular}
}
\end{table}

\subsection{Variance Estimation}\label{sim3} We evaluate the performance of the
variance estimator $\hat{\V}_{N}^{\mvrs}$, defined in Equation~\eqref{eq:3} of
Section~\ref{sec.mseest}. We report the results on for logistic regression and
omit the results for other cases because they exhibit similar patterns.
Figure~\ref{fig:mse_est} displays the means of the estimated MSEs based on
$\hat{\V}_{N}^{\mvrs}$ (dashed lines, labeled ``est'') alongside the empirical
MSEs (solid lines). The estimated MSEs are close to the empirical MSEs,
demonstrating that the proposed variance estimator provides a reliable method
for quantifying the uncertainty of the estimator.

\begin{figure}[H]
    \centering
    \begin{subfigure}{0.45\linewidth}
        \centering
        \includegraphics[width=\linewidth]{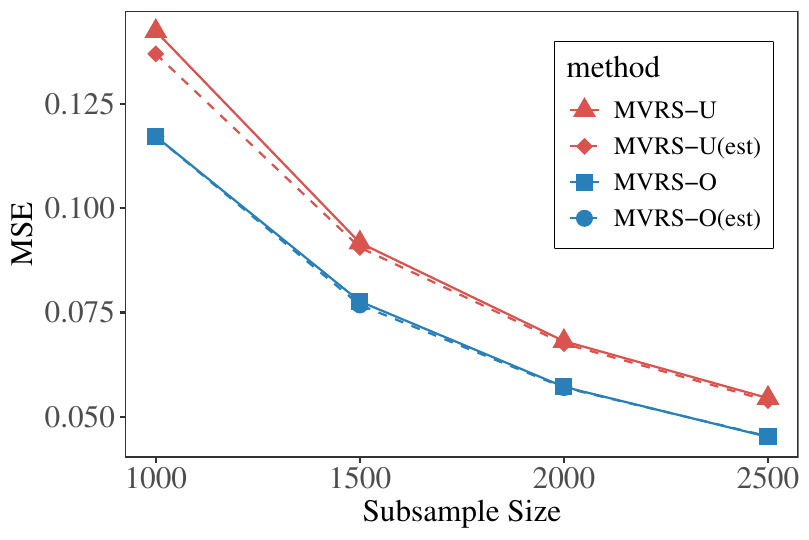}
        \caption{Case 1 (mzNormal)}
    \end{subfigure}
    \hspace{0.05\linewidth}
    \begin{subfigure}{0.45\linewidth}
        \centering
        \includegraphics[width=\linewidth]{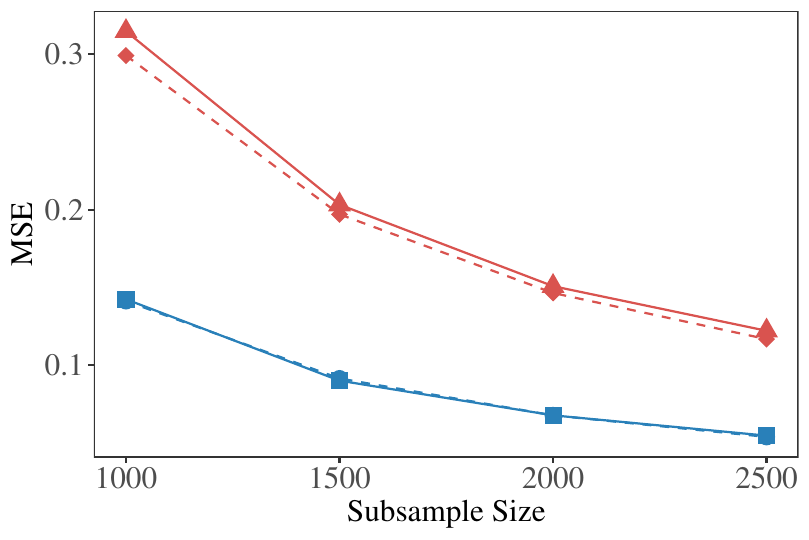}
        \caption{Case 2 (nzNormal)}
    \end{subfigure}
    \vfill
    \begin{subfigure}{0.45\linewidth}
        \centering
        \includegraphics[width=\linewidth]{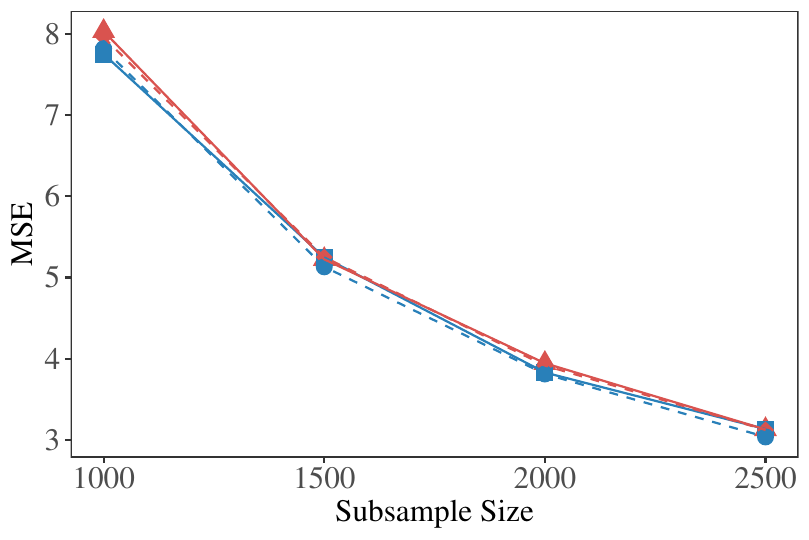}
        \caption{Case 3 (ueNormal)}
    \end{subfigure}
    \hspace{0.05\linewidth}
    \begin{subfigure}{0.45\linewidth}
        \centering
        \includegraphics[width=\linewidth]{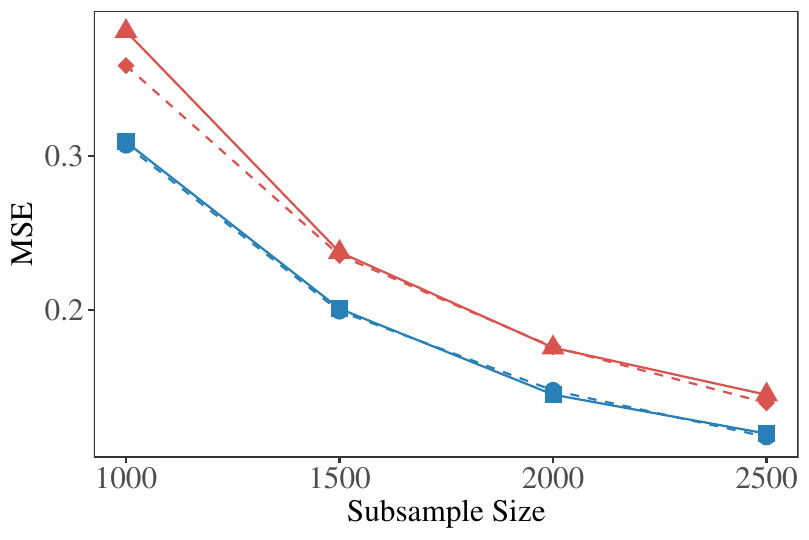}
        \caption{Case 4 (EXP)}
    \end{subfigure}
    \caption{Empirical and Estimated MSEs with different subsample sizes $n$ for logistic regression.}
    \label{fig:mse_est}
\end{figure}

\section{Real-World Data Example}
\label{sec.cas}
 
We evaluate the performance of the proposed MVRS subsampling method on
real-world data using the Supersymmetric (SUSY) benchmark dataset
\citep{susy_279} and the Covertype dataset \citep{covertype_31}. The SUSY
dataset is available on the UCI Machine Learning Repository
(\url{https://archive.ics.uci.edu/ml/datasets/SUSY}), and it contains $5 \times
10^6$ observations with 18 features, consisting of 8 low-level kinematic
properties and 10 high-level derived features. The goal is to distinguish
between signal processes that produce supersymmetric particles and background
processes. The binary version of the Covertype dataset is obtained from the
LIBSVM repository
(\url{https://www.csie.ntu.edu.tw/~cjlin/libsvmtools/datasets/binary.html}), and
it contains 581,012 observations with 54 covariates. The objective is to predict
forest cover types, specifically distinguishing the most prevalent class
(Spruce/Fir) from all others. Among the 54 covariates, 10 are quantitative
cartographic variables—such as elevation, slope, and distances to
hydrology—while the remaining 44 are binary indicators for wilderness areas and
soil types. In our analysis, we utilize the 10 quantitative features for
logistic regression, as the indicator variables are extremely sparse and take
zero values for the vast majority of observations.

We employ a logistic regression model for classification.
The subsample sizes are set to $n \in
\{1000, 1500, 2000, 2500\}$ with a pilot subsample size of $n_0=500$, and the
number of strata is set to $k=50$. We perform the proposed and comparative
methods over $R=1000$ replications and calculate the empirical MSE as
$\operatorname{MSE}=R^{-1}\sum_{r=1}^R\| \htheta_{n,r}-\htheta_N\|^2$, where
$\htheta_N$ denotes the full data estimator. 

Figure~\ref{fig:cas} presents the resulting empirical MSEs. Consistent with the
simulation results in Section~\ref{sim}, MVRS-U and MVRS-O outperform their
corresponding baseline methods (UNIF and OPT). MVRS-U achieves performance
comparable to and even better than OPT, illustrating the practical effectiveness of
the proposed method in real-world scenarios.
Moreover, the results demonstrate that the stratification strategy significantly
enhances estimation efficiency in real-world data analysis. This improvement
arises because real-world datasets seldom follow any assumed statistical model
strictly. Consequently, the design of optimal sampling probabilities can be
affected by model misspecification. The introduction of MVRS improves robustness
by selecting a subsample that is more representative of the full dataset's
underlying structure.

\begin{figure}[H]
    \centering
    \begin{subfigure}{0.45\linewidth}
        \centering
        \includegraphics[width=\linewidth]{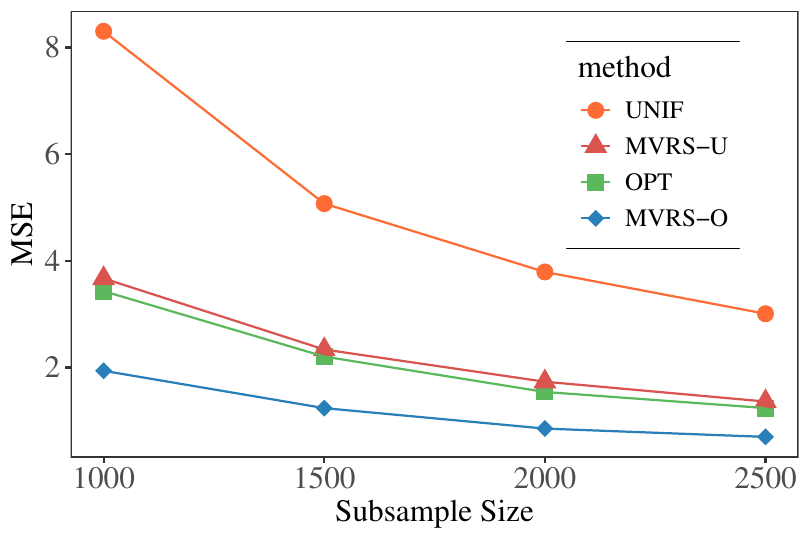}
        \caption{SUSY dataset}
    \end{subfigure}
    \hspace{0.05\linewidth}
    \begin{subfigure}{0.45\linewidth}
        \centering
        \includegraphics[width=\linewidth]{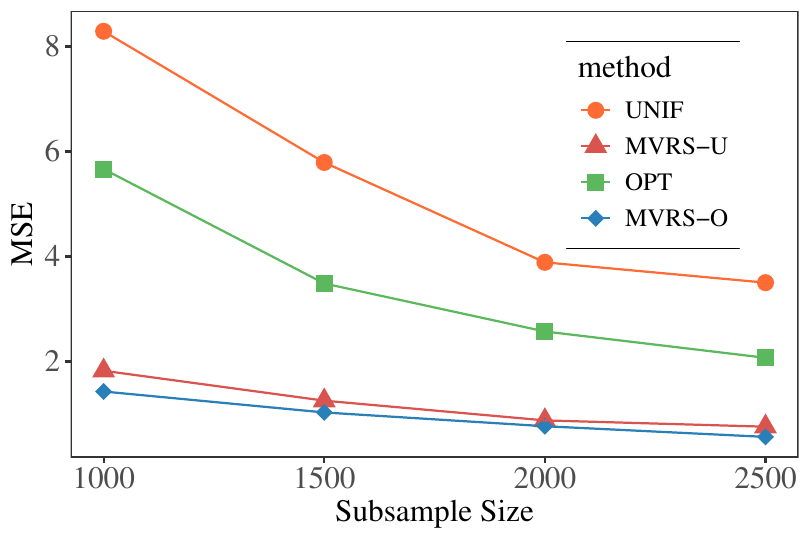}
        \caption{Covertype dataset}
    \end{subfigure}
    \caption{MSEs of different subsample sizes $n$ for case study.}
    \label{fig:cas}
\end{figure}

To provide further information, we calculate the empirical standard error,
$\operatorname{SE}=\sqrt{\operatorname{MSE}}$, and the empirical biases,
$\operatorname{Bias}=R^{-1}\sum_{r=1}^R\htheta_{n,r}-\htheta_N$, for all
components of $\theta$. Table~\ref{tab:cas2} presents the results with $n=1000$
for the Covertype dataset. The empirical biases are much smaller than the
empirical SEs, indicating that the proposed methods are nearly unbiased in terms
of approximating the full data estimator. 

To quantify the efficiency gain of the proposed MVRS methods, we include the SE
ratios relative to the corresponding baseline methods in Table~\ref{tab:cas2}.
Defined as the MVRS SE divided by the baseline SE, a ratio less than 1 indicates
improved efficiency.  While it may appear counter-intuitive that some ratios
exceed 1, this occurs only for components with inherently small SEs. Conversely,
for components with larger SEs, the ratios are consistently below 1.  For
example, MVRS-U and MVRS-O achieve significant efficiency gains for $\theta_{1}$
and $\theta_{8}$.  Because MVRS significantly reduces the error for
high-variance components, the overall MSE is lower than that of the baseline, as
shown in Figure~\ref{fig:cas}.  This aligns with the design of the MVRS scheme,
which targets the direction of the influence function with the maximal variance,
thereby mitigating the dominant sources of error.

\begin{table}[H]
   
\centering
\caption{  Empirical SEs and Biases of individual coefficients with $n=1000$ for
  the Covertype dataset. The SE ratio is calculated as the SE of MVRS divided
  by the SE of the corresponding baseline method.}
\label{tab:cas2}
\resizebox{1\textwidth}{!}{
\begin{tabular}{c|cc c cc c c|cc c cc c c}
    \hline
    \multirow{2}{*}{$\theta_i$}&
  \multicolumn{2}{c}{UNIF}&&\multicolumn{2}{c}{MVRS-U}&& &
  \multicolumn{2}{|c}{OPT}&& \multicolumn{2}{c}{MVRS-O}&&\\
    & SE & Bias && SE & Bias && SE Ratio & SE & Bias && SE & Bias &&SE Ratio\\
    \hline
1 & 4.270 & 0.014 && 0.465 & -0.028 && 0.109 & 3.051 & -0.098 && 0.329 & -0.016 && 0.108 \\
2 & 0.890 & -0.119 && 0.678 & -0.053 && 0.762 & 0.782 & -0.003 && 0.638 & -0.025 && 0.816 \\
3 & 0.025 & 0.004 && 0.022 & -0.007 && 0.870 & 0.031 & 0.002 && 0.026 & -0.002 && 0.835 \\
4 & 0.052 & -0.003 && 0.023 & 0.004 && 0.448 & 0.050 & 0.007 && 0.027 & -0.003 && 0.549 \\
5 & 0.020 & 0.007 && 0.021 & 0.014 && 1.002 & 0.030 & 0.007 && 0.028 & 0.004 && 0.917 \\
6 & 0.014 & 0.000 && 0.013 & -0.010 && 0.968 & 0.016 & -0.005 && 0.017 & 0.004 && 1.051 \\
7 & 0.019 & -0.001 && 0.018 & 0.000 && 0.962 & 0.028 & -0.003 && 0.026 & 0.003 && 0.922 \\
8 & 1.820 & 0.078 && 0.345 & 0.043 && 0.189 & 1.303 & 0.062 && 0.276 & 0.015 && 0.212 \\
9 & 0.478 & 0.021 && 0.149 & 0.028 && 0.312 & 0.340 & 0.030 && 0.130 & 0.023 & &0.381 \\
10 & 0.018 & 0.006 && 0.017 & 0.006 && 0.957 & 0.026 & 0.003 && 0.026 & -0.002 && 1.030 \\
    \hline
\end{tabular}
}
\end{table}

\section{Conclusion and Discussion}
\label{sec.con}
In this paper, we proposed the MVRS strategy to improve estimation efficiency
for existing subsampling distributions in M-estimation with large-scale data.
MVRS can be effectively combined with various subsampling methods to enhance
performance while incurring only an additional linear computational cost. Both
theoretical derivations and numerical experiments demonstrate the superiority of
the proposed method in terms of estimation accuracy. Although we presented MVRS
in the context of subsampling with replacement, the framework can be explicitly
extended to other random subsampling approaches, such as Poisson subsampling.

Several avenues remain for future research. First, we focused on M-estimation
within a parametric framework where the influence function is well-defined for
use as a stratification variable. Extending MVRS to nonparametric and
semiparametric regression problems, where the influence function may not be
readily available, warrants further investigation. Second, while we have applied MVRS
to improve the estimation efficiency of random subsampling, it would be interesting to study its integration with
design-based deterministic subsampling methods, such as IBOSS
\citep{wang2019information, cheng2020information} and orthogonal subsampling
\citep{wang2021orthogonal}. A key challenge in this direction is adjusting MVRS
to avoid the potential bias introduced by using response information for
stratification.   Finally, while MVRS targets the direction of maximum variance,
future work could explore alternative stratification schemes optimized for other
purposes such as better estimation of specific parameters of interest or
predictive performance.

\bibliographystyle{apalike}
\bibliography{ref}

\newpage
\appendix
\setcounter{equation}{0}
\renewcommand{\theequation}{A.\arabic{equation}}

\section{Proof of Theorem~\ref{str.thm}}\label{pf.thm1} We begin the proof of
Theorem~\ref{str.thm} by presenting a series of lemmas. Lemma~\ref{lem1} is a
general result used by Lemma~\ref{lem2} and Lemma~\ref{lem3}. Lemma~\ref{lem2}
states the consistency of the stratified estimator $\htheta_n^{\str}$.
Lemma~\ref{lem3} and Lemma~\ref{lem4} contain two key results for proving the
asymptotic normality of $\htheta_n^{\str}$. We will complete the proof of
Theorem~\ref{str.thm} based on these lemmas.
\begin{lemma}\label{lem1} Let $g(x): \mathbb{R}^d \rightarrow \mathbb{R}$ be a
  function that may be dependent on the full data $\DN$ and satisfies
  $N^{-1}\sumN g(X_i)^2=\Op$. Under Assumption~\ref{asmp05}, as $n\rightarrow
  \infty$ and $N\rightarrow \infty$,
    \begin{equation}\label{lem1.eq1}
        \frac{1}{N^2}\sum_{j=1}^k\sum_{i \in I_j}\frac{\pi_i}{\Pi_j^2} \left\{\frac{\Pi_j}{\pi_i}g(X_i)-\sum_{i \in I_j} g(X_i)\right\}^2=\Op.
    \end{equation}
\end{lemma}
\begin{proof}
    First, we show that the left hand side of \eqref{lem1.eq1} is bounded by
    \begin{equation*}
        \frac{1}{N^2}\sumN\pi_i\left\{\frac{1}{\pi_i}g(X_i)-\sumN g(X_i)\right\}^2.
    \end{equation*}
    According to direct calculation:

\begin{align*}
    &\sum_{j=1}^k\sum_{i \in I_j}\frac{\pi_i}{\Pi_j^2} \left\{\frac{\Pi_j}{\pi_i}g(X_i)-\sum_{i \in I_j} g(X_i)\right\}^2-\sumN\pi_i\left\{\frac{1}{\pi_i}g(X_i)-\sumN g(X_i)\right\}^2\\
    &=\sum_{j=1}^k\sum_{i \in I_j}\frac{\pi_i}{\Pi_j^2} \left\{\frac{\Pi_j}{\pi_i}g(X_i)-\sum_{i \in I_j} g(X_i)\right\}^2-\sum_{j=1}^k\sum_{i \in I_j}\pi_i\left\{\frac{1}{\pi_i}g(X_i)-\sumN g(X_i)\right\}^2\\
    &=\sum_{j=1}^k\sum_{i \in I_j}\pi_i\left[\left\{\frac{1}{\pi_i}g(X_i)-\frac{1}{\Pi_j}\sum_{i \in I_j} g(X_i)\right\}^2-\left\{\frac{1}{\pi_i}g(X_i)-\sumN g(X_i)\right\}^2\right]\\
    &=\sum_{j=1}^k\sum_{i \in I_j}\pi_i\left[\left\{\frac{1}{\Pi_j}\sum_{i \in I_j} g(X_i)\right\}^2+\left\{\sumN g(X_i)\right\}^2-2\left\{\frac{1}{\pi_i}g(X_i)\right\}\left\{\frac{1}{\Pi_j}\sum_{i \in I_j} g(X_i)+\sumN g(X_i)\right\}\right]\\
    &=-\sum_{j=1}^k\frac{1}{\Pi_j}\left\{\sum_{i \in I_j} g(X_i)\right\}^2-\left\{\sumN g(X_i)\right\}^2\\
    &\leq0.
\end{align*}

    On the other hand,
\begin{align*}
    \frac{1}{N^2}\sumN\pi_i\left\{\frac{1}{\pi_i}g(X_i)-\sumN g(X_i)\right\}^2&=\frac{1}{N^2}\sumN\frac{1}{\pi_i}g(X_i)^2-\left\{\oneN\sumN g(X_i)\right\}^2\\
    &\leq\oneN\max_{i}\left(\frac{1}{N\pi_i}\right)\sumN g(X_i)^2\\
    &=\Op,
\end{align*}
where the last equality holds according to Assumption~\ref{asmp05}. Therefore
\eqref{lem1.eq1} holds according to 
\begin{align*}
    0\leq \frac{1}{N^2}\sum_{j=1}^k\sum_{i \in I_j}\frac{\pi_i}{\Pi_j^2} \left\{\frac{\Pi_j}{\pi_i}g(X_i)-\sum_{i \in I_j} g(X_i)\right\}^2\leq 
    \frac{1}{N^2}\sumN\pi_i\left\{\frac{1}{\pi_i}g(X_i)-\sumN g(X_i)\right\}^2=\Op.
\end{align*}
\end{proof}

\begin{lemma}\label{lem2} Under Assumptions~\ref{asmp01}, \ref{asmp02}, and
  \ref{asmp05}, as $n$ and $N$ go to infinity, 
    \begin{equation*}
        \| \htheta_n^{\str}-\htheta_N\| = o_{p|\DN}(1),
    \end{equation*}
where $o_{p|\DN}(1)$ denotes convergence to zero in probability conditional on
$\DN$.
\end{lemma}
\begin{proof}
  Given the full data $\DN$, the randomness of the objective function in
  \eqref{str.est} is solely due to the subsampling process. Based on the
  sampling scheme, $X_{j,i}^*$ is a random sample from $\{X_i\}_{i \in I_j}$
  with selection probabilities $\{\frac{\pi_i}{\Pi_j}\}_{i \in I_{j}}$.
  Therefore for any $\theta \in \Theta$, the conditional expectation and
  variance of 
\begin{equation*}
    \oneN\sum_{j=1}^k\frac{1}{n_j}\sum_{i=1}^{n_j}\frac{\Pi_j}{\pi_{j,i}^*}l(X_{j,i}^*;\theta)
\end{equation*}
given $\DN$ are
\begin{align*}
    \Exp\left\{\oneN\sum_{j=1}^k\frac{1}{n_j}\sum_{i=1}^{n_j}\frac{\Pi_j}{\pi_{j,i}^*}l(X_{j,i}^*;\theta)\Bigg|\DN\right\}
    &=\oneN\sum_{j=1}^k\Exp\left\{\frac{\Pi_j}{\pi_{j,1}^*}l(X_{j,1}^*;\theta)\Bigg|\DN\right\}\\
    &=\oneN\sum_{j=1}^k\sum_{i \in I_j}\frac{\pi_i}{\Pi_j}\frac{\Pi_j}{\pi_i}l(X_i;\theta)=\oneN\sumN l(X_i;\theta)
\end{align*}
and
\begin{align*}
    \Var\left\{\oneN\sum_{j=1}^k\frac{1}{n_j}\sum_{i=1}^{n_j}\frac{\Pi_j}{\pi_{j,i}^*}l(X_{j,i}^*;\theta)\Bigg|\DN\right\}
    &=\frac{1}{N^2}\sum_{j=1}^k\frac{1}{n_j}Var\left\{\frac{\Pi_j}{\pi_{j,1}^*}l(X_{j,1}^*;\theta)\Bigg|\DN\right\}\\
    &=\frac{1}{N^2}\sum_{j=1}^k\frac{1}{n_j}\sum_{i \in I_j}\frac{\pi_i}{\Pi_j}\left\{\frac{\Pi_j}{\pi_i}l(X_i;\theta)-\sum_{i \in I_j} l(X_i;\theta)\right\}^2\\
    &=\frac{1}{nN^2}\sum_{j=1}^k\sum_{i \in I_j}\frac{\pi_i}{\Pi_j^2} \left\{\frac{\Pi_j}{\pi_i}l(X_i;\theta)-\sum_{i \in I_j} l(X_i;\theta)\right\}^2\\
    &=O_{p}\left(n^{-1}\right).
\end{align*}
The last equality above holds because of Assumption~\ref{asmp02} and
Lemma~\ref{lem1} with $g(x)=l(x;\theta)$. Then according to the Chebyshev's
inequality,
\begin{equation}\label{lem2.eq1}
    \oneN\sum_{j=1}^k\frac{1}{n_j}\sum_{i=1}^{n_j}\frac{\Pi_j}{\pi_{j,i}^*}l(X_{j,i}^*;\theta)-\oneN\sumN l(X_i;\theta)=o_{p|\DN}(1).
\end{equation}
Under Assumptions~\ref{asmp01} and \ref{asmp04}, the parameter space $\Theta$ is
compact and the risk function $N^{-1}\sumN l(X_i;\theta)$ has a unique minimum,
as it is continuous and convex. Therefore function
$\oneN\sum_{j=1}^k\frac{1}{n_j}\sum_{i=1}^{n_j}\frac{\Pi_j}{\pi_{j,i}^*}l(X_{j,i}^*;\theta)$
and $\oneN\sumN l(X_i;\theta)$ satisfy the condition of Theorem 5.7 of
\cite{van2000asymptotic}. Then according to the Theorem 5.7 of
\cite{van2000asymptotic} and Assumption~\ref{asmp01} we have proved
Lemma~\ref{lem2}.
\end{proof}

\begin{lemma}\label{lem3} Under Assumptions~\ref{asmp01}, \ref{asmp02},
    \ref{asmp04}, and \ref{asmp05}, as $n$ and $N$ go to infinity,
    \begin{equation}\label{lem3.eq1}
        \oneN\int_0^1\left[\sum_{j=1}^k\frac{\Pi_j}{n_j}\sum_{i=1}^{n_j}\frac{1}{\pi_{j,i}^*}\ddot{l}\{X_{j,i}^*;\htheta_N+\lambda(\htheta_n^{\str}-\htheta_N)\}\right]d\lambda -\oneN\sumN\ddot{l}(X_i;\htheta_N)=o_{p|\DN}(1).
    \end{equation}
\end{lemma}
\begin{proof}
The proof of Lemma~\ref{lem3} is done by examining every element of the matrix.
We consider the $m_1$th row and $m_2$th column of the left side and for
convenience it is still denoted by
$\oneN\int_0^1\left[\sum_{j=1}^k\frac{\Pi_j}{n_j}\sum_{i=1}^{n_j}\frac{1}{\pi_{j,i}^*}\ddot{l}\{X_{j,i}^*;\htheta_N+\lambda(\htheta_n^{\str}-\htheta_N)\}\right]d\lambda
-\oneN\sumN\ddot{l}(X_i;\htheta_N)$. First, we divide the left hand side of
\eqref{lem3.eq1} into two parts as
\begin{align}\label{lem3.eq2}
    &\Bigg|\oneN\int_0^1\left[\sum_{j=1}^k\frac{\Pi_j}{n_j}\sum_{i=1}^{n_j}\frac{1}{\pi_{j,i}^*}\ddot{l}\{X_{j,i}^*;\htheta_N+\lambda(\htheta_n^{\str}-\htheta_N)\}\right]d\lambda -\oneN\sumN\ddot{l}(X_i;\htheta_N)\Bigg| \\ \notag
    &\leq \Bigg|\int_0^1\oneN\left[\sum_{j=1}^k\frac{\Pi_j}{n_j}\sum_{i=1}^{n_j}\frac{1}{\pi_{j,i}^*}[\ddot{l}\{X_{j,i}^*;\htheta_N+\lambda(\htheta_n^{\str}-\htheta_N)\}-\ddot{l}(X_{j,i}^*;\htheta_N)]\right]d\lambda\Bigg|\\ \notag
    &+\Bigg|\oneN\sum_{j=1}^k\frac{\Pi_j}{n_j}\sum_{i=1}^{n_j}\frac{1}{\pi_{j,i}^*}\ddot{l}(X_{j,i}^*;\htheta_N)-\oneN\sumN\ddot{l}(X_i;\htheta_N)\Bigg|.
\end{align}
The first part can be controlled by
\begin{align*}
    &\Bigg|\int_0^1\oneN\sum_{j=1}^k\frac{\Pi_j}{n_j}\sum_{i=1}^{n_j}\frac{1}{\pi_{j,i}^*}[\ddot{l}\{X_{j,i}^*;\htheta_N+\lambda(\htheta_n^{\str}-\htheta_N)\}-\ddot{l}(X_{j,i}^*;\htheta_N)]d\lambda\Bigg| \\
    &\leq \oneN\sum_{j=1}^k\frac{\Pi_j}{n_j}\sum_{i=1}^{n_j}\frac{1}{\pi_{j,i}^*}\Bigg|\ddot{l}\{X_{j,i}^*;\htheta_N+\lambda(\htheta_n^{\str}-\htheta_N)\}-\ddot{l}(X_{j,i}^*;\htheta_N)\Bigg|\\
    &\leq \left[\oneN\sum_{j=1}^k\frac{\Pi_j}{n_j}\sum_{i=1}^{n_j}\frac{1}{\pi_{j,i}^*}c(X_{j,i}^*)\right]\| \htheta_n^{\str}-\htheta_N\|,
\end{align*}
where the last inequality is due to Assumption~\ref{asmp04}. Since $c(x)$
satisfies the condition of Lemma~\ref{lem1}, similar as the proof
of~\eqref{lem2.eq1}, we have
\begin{align*}
    \oneN\sum_{j=1}^k\frac{\Pi_j}{n_j}\sum_{i=1}^{n_j}\frac{1}{\pi_{j,i}^*}c(X_{j,i}^*)-\oneN\sum_{j=1}^N c(X_{i})=o_{p|\DN}(1).
\end{align*}
Then the first part of~\eqref{lem3.eq2} is $o_{p|\DN}(1)$ according to
Lemma~\ref{lem2}. 

For the second part of~\eqref{lem3.eq2}, since $\ddot{l}(x;\htheta_N)$ also
satisfies the condition of Lemma~\ref{lem1}, similar as the proof
of~\eqref{lem2.eq1} we have 
\begin{equation*}
    \Bigg|\oneN\sum_{j=1}^k\frac{\Pi_j}{n_j}\sum_{i=1}^{n_j}\frac{1}{\pi_{j,i}^*}\ddot{l}(X_{j,i}^*;\htheta_N)-\oneN\sumN\ddot{l}(X_i;\htheta_N)\Bigg|=o_{p|\DN}(1).
\end{equation*}
\end{proof}

\begin{lemma}\label{lem4} Under Assumptions~\ref{asmp01}-\ref{asmp05}, as $n$
    and $N$ go to infinity,
    \begin{equation*}\label{lem4.eq1}
        \Phi^{\str}(\htheta_N)^{-\frac{1}{2}}\frac{\sqrt{n}}{N}\sum_{j=1}^k\frac{\Pi_j}{n_j}\sum_{i=1}^{n_j}\frac{1}{\pi_{j,i}^*}\dot{l}(X_{j,i}^*;\htheta_N)\xrightarrow{|\DN} \Nor(0,I),
    \end{equation*}
    where $\xrightarrow{|\DN}$ denotes convergence in distribution conditional
    on $\DN$.
\end{lemma}
\begin{proof}
Let 
\begin{equation*}
    \xi_{j,i}:=\oneN\frac{\Pi_j}{\pi_{j,i}^*}\dot{l}(X_{j,i}^*;\htheta_N)-\oneN\sum_{i \in I_j}\dot{l}(X_i;\htheta_N).
\end{equation*}
Note that $\xi_{j,i}$ are random variables related to the sampling process given
dataset $\DN$. Its conditional expectation and  variance are calculated as
follows:
\begin{equation*}
    \Exp(\xi_{j,i}|\DN)=0.
\end{equation*}
\begin{equation*}
    Var(\xi_{j,i}|\DN)=\frac{1}{N^2}\sum_{i \in I_j} \frac{\pi_i}{\Pi_j}\left\{\frac{\Pi_j}{\pi_i} \dot{l}(X_i;\htheta_N)-\sum_{i \in I_j}\dot{l}(X_i;\htheta_N)\right\}^{\otimes2}=\Op,
\end{equation*}
where the last equality holds according to Assumption~\ref{asmp03} and
Lemma~\ref{lem1} with $g(x)=\dot{l}(x;\htheta_N)$. Although we consider $g(X)
\in \mathbb{R}$ in Lemma~\ref{lem1}, the whole proof still holds for
multi-dimensional $g(X)$. Since $n_j=n\Pi_j$, we know $n_j\rightarrow \infty$ as
$n\rightarrow \infty$. 

First, we check the Lindeberg's condition and establish the asymptotic
normality of $\frac{1}{n_j}\sum_{i=1}^{n_j}\xi_{j,i}$ given $\DN$. For any
$\epsilon>0$,
\begin{align*}
    &\frac{1}{n_j}\sum_{i=1}^{n_j}\Exp\left\{\|\xi_{j,i}\|^2I_{(\|\xi_{j,i}\|>\sqrt{n_j}\epsilon)}\Big|\DN\right\}\\
    &\leq \frac{1}{n_j^{1+\delta/2}\epsilon^{\delta}}\sum_{i=1}^{n_j}\Exp\left\{\|\xi_{j,i}\|^{2+\delta}I_{(\|\xi_{j,i}\|>\sqrt{n_j}\epsilon)}\Big|\DN\right\}\\
    &\leq \frac{1}{n_j^{1+\delta/2}\epsilon^{\delta}}\sum_{i=1}^{n_j}\Exp\left(\|\xi_{j,i}\|^{2+\delta}\Big|\DN\right)\\
    &\leq\frac{1}{n_j^{\delta/2}\epsilon^{\delta}}\Exp\left(\|\xi_{j,1}\|^{2+\delta}\Big|\DN\right)\\
    &\leq\frac{1}{n_j^{\delta/2}\epsilon^{\delta}}\left[ \frac{1}{N^{2+\delta}}\Exp\left\{\frac{\Pi_j^{2+\delta}}{\pi_{j,i}^{*2+\delta}}\|\dot{l}(X_{j,i}^*;\htheta_N)\|^{2+\delta}\Big|\DN\right\}+\frac{1}{N^{2+\delta}}\sum_{i \in I_j}\|\dot{l}(X_i;\htheta_N)\|^{2+\delta}\right]\\
    &\leq\frac{\Pi_j^{1+\delta}}{n_j^{\delta/2}\epsilon^{\delta}}\max_{i}\left(\frac{1}{N\pi_i}\right)^{1+\delta}\oneN\sum_{i \in I_j}\|\dot{l}(X_i;\htheta_N)\|^{2+\delta}+\frac{1}{n_j^{\delta/2}\epsilon^{\delta}N^{2+\delta}}\sum_{i \in I_j}\|\dot{l}(X_i;\htheta_N)\|^{2+\delta}\\
    &=O_{p}\left(n_j^{-\delta/2}\right).
\end{align*}
Denote $Var(\xi_{j,i}|\DN)$ by $\Phi^{\str}_j$. According to the
Lindeberg-Feller Central Theorem, 
\begin{align*}
    \sqrt{n_j}(\Phi_j^{\str})^{-1/2}\frac{1}{n_j}\sum_{i=1}^{n_j}\xi_{j,i}\xrightarrow{|\DN}\Nor(0,I).
\end{align*}
Returning to \eqref{lem4.eq1}, we have
\begin{align*}
    &\frac{\sqrt{n}}{N}\sum_{j=1}^k\frac{\Pi_j}{n_j}\sum_{i=1}^{n_j}\frac{1}{\pi_{j,i}^*}\dot{l}(X_{j,i}^*;\htheta_N)=\sqrt{n}\sum_{j=1}^k\left(\frac{1}{n_j}\sum_{i=1}^{n_j}\xi_{j,i}\right)=\sum_{j=1}^k\frac{1}{\sqrt{\Pi_j}}\sqrt{n_j}\frac{1}{n_j}\sum_{i=1}^{n_j}\xi_{j,i}.
\end{align*}
Since sampling is independent across different strata,
\begin{equation*}
    \Phi^{\str}(\htheta_N)^{-\frac{1}{2}}\frac{\sqrt{n}}{N}\sum_{j=1}^k\frac{\Pi_j}{n_j}\sum_{i=1}^{n_j}\frac{1}{\pi_{j,i}^*}\dot{l}(X_{j,i}^*;\htheta_N)\xrightarrow{|\DN} \Nor(0,I).
\end{equation*}
\end{proof}

\begin{proof}[\bf Proof of Theorem~\ref{str.thm}]
According to Lemma~\ref{lem2}, $\htheta_n^{\str}$ is consistent to $\htheta_N$.
Therefore, we can apply Taylor's expansion to the target function in
\eqref{str.est} at point $\htheta_N$ as 
\begin{align*}
    0&=\oneN\sum_{j=1}^{k}\frac{\Pi_j}{n_j}\sum_{i=1}^{n_j} \frac{1}{\pi_{j,i}^*}\dot{l}(X_{j,i}^*; \htheta_n^{\str})\\
    &=\oneN\sum_{j=1}^{k}\frac{\Pi_j}{n_j}\sum_{i=1}^{n_j} \frac{1}{\pi_{j,i}^*}\dot{l}(X_{j,i}^*; \htheta_N)\\
    &+\left[\oneN\int_0^1\sum_{j=1}^k\frac{\Pi_j}{n_j}\sum_{i=1}^{n_j}\frac{1}{\pi_{j,i}^*}\ddot{l}\{X_{j,i}^*;\htheta_N+\lambda(\htheta_n^{\str}-\htheta_N)\}d\lambda\right](\htheta_n^{\str}-\htheta_N)\\
    &=\oneN\sum_{j=1}^{k}\frac{\Pi_j}{n_j}\sum_{i=1}^{n_j} \frac{1}{\pi_{j,i}^*}\dot{l}(X_{j,i}^*; \htheta_N)+\left\{\oneN\sumN\ddot{l}(X_i;\htheta_N)+o_{p|\DN}(1)\right\}(\htheta_n^{\str}-\htheta_N).
\end{align*}
Here, the first equality holds according to the definition of $\htheta_n^{\str}$
in \eqref{str.est}, and the last equality holds due to Lemma~\ref{lem3}. Based
on Assumption~\ref{asmp03}, rearranging the above equation gives
\begin{align*}
    \htheta_n^{\str}-\htheta_N=-\left\{\oneN\sumN\ddot{l}(X_i;\htheta_N)+o_{p|\DN}(1)\right\}^{-1}\oneN\sum_{j=1}^{k}\frac{\Pi_j}{n_j}\sum_{i=1}^{n_j} \frac{1}{\pi_{j,i}^*}\dot{l}(X_{j,i}^*; \htheta_N).
\end{align*}
According to Lemma~\ref{lem4} and Slutsky's theorem, we proved that 
\begin{align*}
    \htheta_n^{\str}-\htheta_N \xrightarrow{|\DN} \Nor\left(0, \V_{N}^{\str}\right),
\end{align*}
which means for any $x$
\begin{align*}
    \Pr\left\{(\V_{N}^{\str})^{-\frac{1}{2}}\sqrt{n}(\htheta_n^{\str}-\htheta_N)\leq x\Big|\DN\right\}\rightarrow \Phi(x),
\end{align*}
where $\Phi(x)$ is the cumulative distribution function of standard multivariate
normal distribution. Since the conditional probability is a bounded random
variable, according to the bounded convergence theorem, we have
\begin{align*} 
    &\Pr\left\{(\V_{N}^{\str})^{-\frac{1}{2}}\sqrt{n}(\htheta_n^{\str}-\htheta_N)\leq x\right\}\\
    &=\Exp\left[\Pr\left\{(\V_{N}^{\str})^{-\frac{1}{2}}\sqrt{n}(\htheta_n^{\str}-\htheta_N)\leq x\Big|\DN\right\}\right]\rightarrow \Phi(x),
\end{align*}
which finishes the proof of Theorem~\ref{str.thm}.
\end{proof}

\section{Proof of Theorem~\ref{str.thm2}}\label{pf.thm2} Now we give the proof
of Theorem~\ref{str.thm2}, which relies on the properties of conditional
variance.
\begin{proof}
According to Proposition~\ref{sub.lem} and Theorem~\ref{str.thm}, the asymptotic
variance matrices can be expressed as
\begin{align*}
    &\V_{N}^{\str}=\oneN\Exp_{\Q_N}\left[
    \Var_{\Q_N}\left\{\frac{\ud\Pr_N}{\ud\Q_N}
      \varphi(X;\theta)\Bigg|S_A\right\}\right], 
      \V_{N}^{\sub}=
    \oneN\Var_{\Q_N}\left\{\frac{\ud\Pr_N}{\ud\Q_N}\varphi(X;\theta)\right\}.
\end{align*}
By the properties of conditional variance,
\begin{align}
    \V_{N}^{\str}-\V_{N}^{\sub}&=\oneN\Exp_{\Q_N}\left[
    \Var_{\Q_N}\left\{\frac{\ud\Pr_N}{\ud\Q_N}
      \varphi(X;\theta)\Bigg|S_A\right\}\right]
    -\oneN\Var_{\Q_N}\left\{\frac{\ud\Pr_N}{\ud\Q_N}\varphi(X;\theta)\right\}\notag\\
    \label{thm2.eq1}
    &=-\oneN\Var_{\Q_N}\left[\Exp_{\Q_N}\left\{\frac{\ud\Pr_N}{\ud\Q_N}\varphi(X;\htheta_N)\Bigg|S_A\right\}\right].
\end{align}
Therefore the MSE reduction of $\htheta_n^{\str}$ can be guaranteed. For a
clearer view of the difference between these two asymptotic variance matrices, we
further calculate~\eqref{thm2.eq1}. Note that
\begin{align*}
\Exp_{\Q_N}\left\{\frac{\ud\Pr_N}{\ud\Q_N}\varphi(X;\htheta_N)\Bigg|S_A=j\right\}
=\frac{1}{\Q_N(S \in A_j)}\int \varphi(X;\htheta_N)\mathbb{I}(S \in A_j)\ud\Pr_N.
\end{align*}
We have
\begin{align*}
\Exp_{\Q_N}\left[\Exp_{\Q_N}\left\{\frac{\ud\Pr_N}{\ud\Q_N}\varphi(X;\htheta_N)\Bigg|S_A\right\}\right]
&=\sum_{j=1}^{k}\Q_N(S \in A_j)\Exp_{\Q_N}\left\{\frac{\ud\Pr_N}{\ud\Q_N}\varphi(X;\htheta_N)\Bigg|S_A=j\right\}\\
&=\sum_{j=1}^{k}\int \varphi(X;\htheta_N)\mathbb{I}(S \in A_j)\ud\Pr_N=\Exp_{\Pr_N}\{\varphi(X;\htheta_N)\},
\end{align*}
\begin{align*}
\Exp_{\Q_N}\left[\Exp^2_{\Q_N}\left\{\frac{\ud\Pr_N}{\ud\Q_N}\varphi(X;\htheta_N)\Bigg|S_A\right\}\right]
&=\sum_{j=1}^{k}\Q_N(S \in A_j)\Exp^2_{\Q_N}\left\{\frac{\ud\Pr_N}{\ud\Q_N}\varphi(X;\htheta_N)\Bigg|S_A=j\right\}\\
&=\sum_{j=1}^{k}\frac{1}{\Q_N(S \in A_j)}\left\{\int \varphi(X;\htheta_N)\mathbb{I}(S \in A_j)\ud\Pr_N\right\}^2.
\end{align*}
On the other hand,
\begin{align*}
\Exp_{\Pr_N}\left\{\varphi(X;\htheta_N)\Bigg|S_A=j\right\}
=\frac{1}{\Pr_N(S \in A_j)}\int \varphi(X;\htheta_N)\mathbb{I}(S \in A_j)\ud\Pr_N
\end{align*}
and
\begin{align*}
\Exp_{\Q_N}\left[\frac{\ud\Pr_N}{\ud\Q_N}\Exp_{\Pr_N}\left\{\varphi(X;\htheta_N)\Bigg|S_A\right\}\right]
&=\Exp_{\Pr_N}\left[\Exp_{\Pr_N}\left\{\varphi(X;\htheta_N)\Bigg|S_A\right\}\right]\\
&=\sum_{j=1}^{k}\int \varphi(X;\htheta_N)\mathbb{I}(S \in A_j)\ud\Pr_N=\Exp_{\Pr_N}\{\varphi(X;\htheta_N)\},
\end{align*}
\begin{align*}
&\Exp_{\Q_N}\left[\frac{d^2\Pr_N}{d^2\Q_N}\Exp^2_{\Pr_N}\left\{\varphi(X;\htheta_N)\Bigg|S_A\right\}\right]
=\sum_{j=1}^{k}\Q_N(S \in A_j)\frac{\Pr_N^2(S \in A_j)}{\Q_N^2(S \in A_j)}\Exp^2_{\Pr_N}\left\{\varphi(X;\htheta_N)\Bigg|S_A=j\right\}\\
&=\sum_{j=1}^{k}\Q_N(S \in A_j)\frac{\Pr_N^2(S \in A_j)}{\Q_N^2(S \in A_j)}\left\{\frac{1}{\Pr_N(S \in A_j)}\int \varphi(X;\htheta_N)\mathbb{I}(S \in A_j)\ud\Pr_N\right\}^2\\
&=\sum_{j=1}^{k}\frac{1}{\Q_N(S \in A_j)}\left\{\int \varphi(X;\htheta_N)\mathbb{I}(S \in A_j)\ud\Pr_N\right\}^2.
\end{align*}
Therefore
\begin{align*}
    &\Var_{\Q_N}\left[\Exp_{\Q_N}\left\{\frac{\ud\Pr_N}{\ud\Q_N}\varphi(X;\htheta_N)\Bigg|S_A\right\}\right]\\
    &=\Exp_{\Q_N}\left[\Exp^2_{\Q_N}\left\{\frac{\ud\Pr_N}{\ud\Q_N}\varphi(X;\htheta_N)\Bigg|S_A\right\}\right]-\Exp^2_{\Q_N}\left[\Exp_{\Q_N}\left\{\frac{\ud\Pr_N}{\ud\Q_N}\varphi(X;\htheta_N)\Bigg|S_A\right\}\right]\\
    &=\sum_{j=1}^{k}\frac{1}{\Q_N(S \in A_j)}\left\{\int \varphi(X;\htheta_N)\mathbb{I}(S \in A_j)\ud\Pr_N\right\}^2-\Exp^2_{\Pr_N}\{\varphi(X;\htheta_N)\}\\
    &=\Exp_{\Q_N}\left[\frac{d^2\Pr_N}{d^2\Q_N}\Exp^2_{\Pr_N}\left\{\varphi(X;\htheta_N)\Bigg|S_A\right\}\right]-\Exp^2_{\Q_N}\left[\frac{\ud\Pr_N}{\ud\Q_N}\Exp_{\Pr_N}\left\{\varphi(X;\htheta_N)\Bigg|S_A\right\}\right]\\
    &=\Var_{\Q_N}\left[{\frac{\ud\Pr_N}{\ud\Q_N}}\Exp_{\Pr_N}\left\{\varphi(X;\htheta_N)\Big|S_A\right\}\right].
\end{align*}
Finally we have
\begin{align*}
    \V_{N}^{\str}-\V_{N}^{\sub}
    &=-\oneN\Var_{\Q_N}\left[\Exp_{\Q_N}\left\{\frac{\ud\Pr_N}{\ud\Q_N}\varphi(X;\htheta_N)\Bigg|S_A\right\}\right]\\
    &=-\oneN\Var_{\Q_N}\left[{\frac{\ud\Pr_N}{\ud\Q_N}}\Exp_{\Pr_N}\left\{\varphi(X;\htheta_N)\Big|S_A\right\}\right] \leq 0.
\end{align*}
\end{proof}

\section{Proof of \eqref{diff.2}}\label{pf.diff.2}

\begin{proof}
Since
\begin{align*}
    &\Var_{\Q_N}\left\{\frac{\ud\Pr_N}{\ud\Q_N}\varphi(X;\htheta_N)\Bigg|S_A\right\}
    -\Exp_{\Q_N}\left[\Var_{\Q_N}\left\{\frac{\ud\Pr_N}{\ud\Q_N}\varphi(X;\htheta_N)\Bigg|S'_A,S_A\right\}\Bigg|S_A\right]\\
    &=\Var_{\Q_N}\left[\Exp_{\Q_N}\left\{\frac{\ud\Pr_N}{\ud\Q_N}\varphi(X;\htheta_N)\Bigg|S'_A,S_A\right\}\Bigg|S_A\right].
\end{align*}
Take the expectation on both sides with respect to $S_A$, we have
\begin{align*}
  &\Exp_{\Q_N}\left[\Var_{\Q_N}\left\{\frac{\ud\Pr_N}{\ud\Q_N}\varphi(X;\htheta_N)\Bigg|S_A\right\}\right]
  -\Exp_{\Q_N}\left[\Var_{\Q_N}\left\{\frac{\ud\Pr_N}{\ud\Q_N}\varphi(X;\htheta_N)\Bigg|S'_A,S_A\right\}\right]\\
  &=\Exp_{\Q_N}\left[\Var_{\Q_N}\left\{\frac{\ud\Pr_N}{\ud\Q_N}\Exp_{\Q_N}\left(\varphi(X;\htheta_N)\Big|S'_A,S_A\right)\Bigg|S_A\right\}\right].
\end{align*}
Therefore,
\begin{align*}
  &\V_{N}^{\str'}(\htheta_N)-\V_{N}^{\str}\\
  &=\oneN\Exp_{\Q_N}\left[\Var_{\Q_N}\left\{\frac{\ud\Pr_N}{\ud\Q_N}\varphi(X;\htheta_N)\Bigg|S'_A\right\}\right]
  -\oneN\Exp_{\Q_N}\left[\Var_{\Q_N}\left\{\frac{\ud\Pr_N}{\ud\Q_N}\varphi(X;\htheta_N)\Bigg|S_A\right\}\right]\\
  &=\oneN\Exp_{\Q_N}\left[\Var_{\Q_N}\left\{\frac{\ud\Pr_N}{\ud\Q_N}\varphi(X;\htheta_N)\Bigg|S'_A,S_A\right\}\right]
  -\oneN\Exp_{\Q_N}\left[\Var_{\Q_N}\left\{\frac{\ud\Pr_N}{\ud\Q_N}\varphi(X;\htheta_N)\Bigg|S_A\right\}\right]\\
  &=-\oneN\Exp_{\Q_N}\left[\Var_{\Q_N}\left\{\frac{\ud\Pr_N}{\ud\Q_N}\Exp_{\Q_N}\left(\varphi(X;\htheta_N)\Big|S'_A,S_A\right)\Bigg|S_A\right\}\right].
\end{align*}
It proves that the variance of the estimator decreases as the number of strata
increases.
\end{proof}

\section{Proof of Theorem~\ref{alg.thm}}
\begin{proof}
  Note that the calculated sampling probabilities $\tilde \pi_i$ may depend on
the pilot subsamples, denoted as $\D_{n_0}$. Similarly to the proof of
Theorem~\ref{str.thm}, we establish the conditional asymptotic normality of
$\hat\theta_n^{\mvrs}$ as
  \begin{equation*}
    \sqrt{n}(\tilde\V_{N}^{\mvrs})^{-1/2}(\hat\theta_n^{\mvrs}-\htheta_N)
    \xrightarrow{\mid \DN,\D_{n_0}} \Nor(0, \I),
  \end{equation*}
  in distribution as $N, n, n_0 \rightarrow \infty$, where
  \begin{align*}
\tilde\V_{N}^{\mvrs}
    &=\Exp_{\tilde \Q_N}\left[
    \Var_{\tilde \Q_N}\bigg\{\frac{\ud\Pr_N}{\ud\tilde \Q_N}
      \varphi(X;\htheta_N)\bigg|\tilde S^{\mvrs}_{\tilde A^{\mvrs}}\bigg\}\right]\\
    &=\oneN\sum_{j=1}^k\sum_{i \in \tilde I^{\mvrs}_j}
      \frac{\tilde \pi_i}{N\tilde \Pi_j^2}
    \bigg\{\frac{\tilde \Pi_j}{\tilde \pi_i}\varphi(X_i;\htheta_N)
      -\sum_{i \in \tilde I^{\mvrs}_j}\varphi(X_i;\htheta_N)\bigg\}^{\otimes2}.
\end{align*}
Here, $\tilde\Q_N=\sumN \tilde\pi_i \delta_{X_i}$ denotes the weighted empirical
  measure based on the pilot estimator. The stratification variable $\tilde S^{\mvrs}_{\tilde A^{\mvrs}}$ is defined as $\tilde S^{\mvrs}_{\tilde A^{\mvrs}_j}=j\mathbb{I}_{\tilde A^{\mvrs}_j}(\tilde S^{\mvrs})$, where
  $\tilde A^{\mvrs}_j=(\tilde S^{\mvrs}_{(j-1)}, \tilde S^{\mvrs}_{(j)}]$ for
  $j=1, ..., k$.

 For any $x$,
\begin{align*}
    \Pr\left\{(\tilde\V_{N}^{\mvrs})^{-\frac{1}{2}}\sqrt{n}(\hat\theta_n^{\mvrs}-\htheta_N)\leq x\Big|\DN, \D_{n_0}\right\}\rightarrow \Phi(x),
\end{align*}
where $\Phi(x)$ is the cumulative distribution function of the standard
multivariate normal distribution. Since the conditional probability is a bounded
random variable, according to the bounded convergence theorem, we have
\begin{align*} 
    &\Pr\left\{(\tilde\V_{N}^{\str})^{-\frac{1}{2}}\sqrt{n}(\htheta_n^{\str}-\htheta_N)\leq x\right\}\\
    &=\Exp_{\DN}\left[\Pr\left\{(\tilde\V_{N}^{\str})^{-\frac{1}{2}}\sqrt{n}(\htheta_n^{\str}-\htheta_N)\leq x\Big|\DN\right\}\right]\\
    &=\Exp_{\DN}\left(\Exp_{\D_{n_0}}\left[\Pr\left\{(\tilde\V_{N}^{\str})^{-\frac{1}{2}}\sqrt{n}(\htheta_n^{\str}-\htheta_N)\leq x\Big|\DN, \D_{n_0}\right\}\right] \right)
    \rightarrow \Phi(x).
\end{align*}
Therefore, we have
\begin{equation*}
    \sqrt{n}(\tilde\V_{N}^{\mvrs})^{-1/2}(\hat\theta_n^{\mvrs}-\htheta_N)
    \rightarrow \Nor(0, \I).
  \end{equation*}
We rewrite the asymptotic variance of $\hat\theta_n^{\mvrs}$ as 
\begin{align}
\label{eq.vmvrs}
\tilde\V_{N}^{\mvrs}
=&\oneN\sum_{j=1}^k\sum_{i \in \tilde I^{\mvrs}_j}\frac{1}{N\tilde \pi_i}\varphi^2(X_i;\htheta_N)+\oneN\sum_{j=1}^k\sum_{i \in \tilde I^{\mvrs}_j}\frac{1}{\tilde \Pi_j}\varphi(X_i;\htheta_N)\oneN\sum_{i \in
   \tilde I^{\mvrs}_j}\varphi(X_i;\htheta_N) \notag\\
&+\sum_{j=1}^k\left\{\oneN\sum_{i \in \tilde I^{\mvrs}_j}\varphi(X_i;\htheta_N)\right\}^2.
\end{align}
We analyze these three terms separately. 

First, we focus on the convergence of
\begin{align*}
  \oneN\sum_{i \in \tilde I^{\mvrs}_j}\varphi(X_i;\htheta_N)-\oneN\sum_{i \in I^{\mvrs}_j}\varphi(X_i;\htheta_N).
\end{align*}  
Define
\begin{align*}
  &L_j=[\min\{\tilde S^{\mvrs}_{(j-1)}, S^{\mvrs}_{(j-1)}\}, \max\{\tilde S^{\mvrs}_{(j-1)}, S^{\mvrs}_{(j-1)}\}],\\
  &U_j=[\min\{\tilde S^{\mvrs}_{(j)}, S^{\mvrs}_{(j)}\}, \max\{\tilde S^{\mvrs}_{(j)}, S^{\mvrs}_{(j)}\}].
\end{align*}
We have
\begin{align*}
&\frac{1}{N}\sumN\{\mI_{\tilde A^{\mvrs}_j}(\tilde S^{\mvrs}_i)-\mI_{A^{\mvrs}_j}(\tilde S^{\mvrs}_i)\}\|\varphi(X_i;\htheta_N)\|\\
&\leq \sqrt{\oneN\sumN\|\varphi(X_i;\htheta_N)\|^2}
  \sqrt{\oneN\sumN\{\mI_{\tilde A^{\mvrs}_j}(\tilde S^{\mvrs}_i)-\mI_{A^{\mvrs}_j}(\tilde S^{\mvrs}_i)\}^2}\\
&=\sqrt{\oneN\sumN\|\varphi(X_i;\htheta_N)\|^2}
  \sqrt{\oneN\sumN\{\mI_{L_j}(\tilde S^{\mvrs}_i)+\mI_{U_j}(\tilde S^{\mvrs}_i)\}}\\
&=\sqrt{\oneN\sumN\|\varphi(X_i;\htheta_N)\|^2}
  \sqrt{\Pr(\tilde S^{\mvrs} \in L_j)+\Pr(\tilde S^{\mvrs} \in U_j)+\op},
\end{align*}
where the last equation holds due to the Glivenko-Cantelli theorem
\citep{shorack2009empirical}. 

Then we consider the intervals $L_j$ and $U_j$. Denote the stratification
variable with the true parameter by $S^{\mvrs}_0=\u\tp\varphi(X;\theta)$. Since 
\begin{align*}
\|\varphi(X;\htheta_{n_0})-\varphi(X;\theta)\|\leq \sup_{\theta' \in [\hat{\theta}_{n_0},\theta]}\|\dot{\varphi}(X;\theta')\|\|\htheta_{n_0}-\theta\|,
\end{align*}
we have $\mid\tilde S^{\mvrs}-S_0^{\mvrs}\mid=\op$ according to
Assumption~\ref{asmp04} and the fact that $\|\htheta_{n_0}-\theta\|=\op$ as $N,
n_0\rightarrow \infty$. Since the density function $S_0^{\mvrs}$ is
positive at the $j/k$-quantiles (for $j=1, ..., k-1$), the $j/k$-quantiles of $\tilde S^{\mvrs}$
converge to those of $S_0^{\mvrs}$ and the sample quantiles converge to the
population quantiles for $\tilde S^{\mvrs}$. Therefore we have $\tilde S^{\mvrs}_{(j)}$ converges to the $j/k$-quantile of $S_0^{\mvrs}$ for $j=1, ...,
k-1$. Similarly, we have $S^{\mvrs}_{(j)}$ also converges to the $j/k$-quantile
of $S_0^{\mvrs}$ for $j=1, ..., k-1$. Thus the widths of the intervals $L_j$ and
$U_j$ converge to $0$ in probability.

Since $\mid\tilde S^{\mvrs}-S_0^{\mvrs}\mid=o_{P}(1)$, we have  
\begin{align*}
    \Pr(\tilde S^{\mvrs} \in L_j)=\Pr(S_0^{\mvrs} \in L_j)+\op=\op+\op
\end{align*}
under the assumption that the distribution function of $S_0^{\mvrs}$ is
continuous at the $j/k$-quantiles for $j=1, ..., k-1$. Additionally with
Assumption~\ref{asmp03}, we have proved that 
\begin{align*}
\frac{1}{N}\sumN\{\mI_{\tilde A^{\mvrs}_j}(\tilde S^{\mvrs}_i)-\mI_{A^{\mvrs}_j}(\tilde S^{\mvrs}_i)\}\|\varphi(X_i;\htheta_N)\|=\op.
\end{align*}
Therefore
\begin{align}
  \frac{1}{N}\sum_{i \in \tilde I^{\mvrs}_j}\varphi(X_i;\htheta_N)
  &=\frac{1}{N}\sumN\mI_{\tilde A^{\mvrs}_j}(\tilde S^{\mvrs}_i)\varphi(X_i;\htheta_N)\notag\\
  &=\frac{1}{N}\sumN\mI_{A^{\mvrs}_j}(\tilde S^{\mvrs}_i)\varphi(X_i;\htheta_N)+\op.\label{eq.D.1}
\end{align}
Consider
\begin{align*}
&\frac{1}{N}\sumN\{\mI_{A^{\mvrs}_j}(\tilde S^{\mvrs}_i)-\mI_{A^{\mvrs}_j}(S^{\mvrs}_i)\}\|\varphi(X_i;\htheta_N)\|\\
&\leq \sqrt{\oneN\sumN\|\varphi(X_i;\htheta_N)\|^2}
  \sqrt{\oneN\sumN\{\mI_{A^{\mvrs}_j}(\tilde S^{\mvrs}_i)-\mI_{A^{\mvrs}_j}(S^{\mvrs}_i)\}^2}\\
&\leq \sqrt{\oneN\sumN\|\varphi(X_i;\htheta_N)\|^2}\\
&\times\sqrt{\oneN\sumN\{\mI_{A^{\mvrs}_j}(\tilde S^{\mvrs}_i)-\mI_{A^{\mvrs}_j}(S^{\mvrs}_{0,i})\}^2+\oneN\sumN\{\mI_{A^{\mvrs}_j}(S^{\mvrs}_i)-\mI_{A^{\mvrs}_j}(S^{\mvrs}_{0,i})\}^2}
\end{align*}
Since we have $\mid\tilde S^{\mvrs}-S_0^{\mvrs}\mid=\op$ and $\mid
S^{\mvrs}-S_0^{\mvrs}\mid=\op$, the last two terms are both $\op$ according to
the assumption that the distribution function of $S_0^{\mvrs}$ is continuous at
the $j/k$-quantiles for $j=1, ..., k-1$.

Therefore 
\begin{align}
  \frac{1}{N}\sum_{i \in \tilde I^{\mvrs}_j}\varphi(X_i;\htheta_N)
  &=\frac{1}{N}\sumN\mI_{A^{\mvrs}_j}(\tilde S^{\mvrs}_i)\varphi(X_i;\htheta_N)+\op\notag\\
  \label{eq.D.2}
  &=\frac{1}{N}\sumN\mI_{A^{\mvrs}_j}(S^{\mvrs}_i)\varphi(X_i;\htheta_N)+\op\\ 
  \label{term3}
  &= \frac{1}{N}\sum_{i \in I^{\mvrs}_j}\varphi(X_i;\htheta_N)+\op.
\end{align}
Now we have proved that the last term of~\eqref{eq.vmvrs} converges. 

Then consider
\begin{align*}
&\frac{1}{N}\sumN\{\mI_{A^{\mvrs}_j}(\tilde S^{\mvrs}_i)-\mI_{A^{\mvrs}_j}(S^{\mvrs}_i)\}\left\|\frac{1}{N\tilde \pi_i}\varphi(X_i;\htheta_N)^{\otimes2}\right\|\\
&\leq \max_{1 \leq i\leq N}\left| \frac{1}{N\tilde \pi_i} \right| \sqrt{ \oneN\sumN\left\|\varphi(X_i;\htheta_N)\right\|^4}
  \sqrt{\oneN\sumN\{\mI_{A^{\mvrs}_j}(\tilde S^{\mvrs}_i)-\mI_{A^{\mvrs}_j}(S^{\mvrs}_i)\}^2}.
\end{align*}
According to Assumption~\ref{asmp05} and the fact that
$\Exp\{\varphi^4(X;\theta)\} < \infty$, we have 
\begin{align*}
\max_{1 \leq i\leq N}\left| \frac{1}{N\tilde \pi_i} \right| \sqrt{ \oneN\sumN\left\|\varphi(X_i;\htheta_N)\right\|^4}=O_P(1).
\end{align*}
Therefore similar to~\eqref{eq.D.1} and~\eqref{eq.D.2} we proved that for the first term of~\eqref{eq.vmvrs}, 
\begin{align*}
\oneN\sum_{i \in \tilde I^{\mvrs}_j}\frac{1}{N\tilde \pi_i}\varphi^2(X_i;\htheta_N)
&=\frac{1}{N}\sumN\mI_{\tilde A^{\mvrs}_j}(\tilde S^{\mvrs}_i)\frac{1}{N\tilde \pi_i}\varphi^2(X_i;\htheta_N)\\
&=\frac{1}{N}\sumN\mI_{A^{\mvrs}_j}(\tilde S^{\mvrs}_i)\frac{1}{N\tilde \pi_i}\varphi^2(X_i;\htheta_N)+\op\\
&=\frac{1}{N}\sumN\mI_{A^{\mvrs}_j}(S^{\mvrs}_i)\frac{1}{N\tilde \pi_i}\varphi^2(X_i;\htheta_N)+\op\\
&= \frac{1}{N}\sum_{i \in I^{\mvrs}_j}\frac{1}{N\tilde \pi_i}\varphi^2(X_i;\htheta_N)+\op.
\end{align*}
Moreover, since $\mid \tilde \pi_i-\pi_i\mid =o_p(1/N)$, we have
\begin{align}
\oneN\sum_{i \in \tilde I^{\mvrs}_j}\frac{1}{N\tilde \pi_i}\varphi^2(X_i;\htheta_N)
&= \frac{1}{N}\sum_{i \in I^{\mvrs}_j}\frac{1}{N\tilde \pi_i}\varphi^2(X_i;\htheta_N)+\op\notag\\
\label{eq.D.3}
&=\frac{1}{N}\sum_{i \in I^{\mvrs}_j}\frac{1}{N \pi_i}\varphi^2(X_i;\htheta_N)+\op
\end{align}
under Assumption~\ref{asmp05}. The convergence of the first term
of~\eqref{eq.vmvrs} has been proved.

Finally we consider
\begin{align*}
&\frac{1}{N}\sumN\{\mI_{A^{\mvrs}_j}(\tilde S^{\mvrs}_i)-\mI_{A^{\mvrs}_j}(S^{\mvrs}_i)\}\left\|\frac{1}{\tilde \Pi_j}\varphi(X_i;\htheta_N)\right\|\\
&\leq \max_{1 \leq j\leq k}\left| \frac{1}{\tilde \Pi_j} \right|\sqrt{ \oneN\sumN\left\|\varphi(X_i;\htheta_N)\right\|^2}
  \sqrt{\oneN\sumN\{\mI_{A^{\mvrs}_j}(\tilde S^{\mvrs}_i)-\mI_{A^{\mvrs}_j}(S^{\mvrs}_i)\}^2}.
\end{align*}
According to Assumption~\ref{asmp05}, we have that $\max_{1\leq j \leq k}\tilde \Pi^{-1}_j=O_P(1)$. Therefore 
\begin{align*}
\max_{1 \leq j\leq k}\left| \frac{1}{\tilde \Pi_j} \right|\sqrt{ \oneN\sumN\left\|\varphi(X_i;\htheta_N)\right\|^2}=O_P(1).
\end{align*}
Then similar to~\eqref{eq.D.1} and~\eqref{eq.D.2} we proved that for the second term of~\eqref{eq.vmvrs},
\begin{align}
\oneN\sum_{i \in \tilde I^{\mvrs}_j}\frac{1}{\tilde \Pi_j}\varphi(X_i;\htheta_N)
&=\frac{1}{N}\sumN\mI_{\tilde A^{\mvrs}_j}(\tilde S^{\mvrs}_i)\frac{1}{\tilde \Pi_j}\varphi(X_i;\htheta_N)\notag\\
&=\frac{1}{N}\sumN\mI_{A^{\mvrs}_j}(\tilde S^{\mvrs}_i)\frac{1}{\tilde \Pi_j}\varphi(X_i;\htheta_N)+\op\notag\\
&=\frac{1}{N}\sumN\mI_{A^{\mvrs}_j}(S^{\mvrs}_i)\frac{1}{\tilde \Pi_j}\varphi(X_i;\htheta_N)+\op\notag\\
&=\oneN\sum_{i \in I^{\mvrs}_j}\frac{1}{\tilde \Pi_j}\varphi(X_i;\htheta_N)+\op\notag\\
\label{term2}
&=\oneN\sum_{i \in I^{\mvrs}_j}\frac{1}{\Pi_j}\varphi(X_i;\htheta_N)+\op.
\end{align}
The last equation can be proved similarly to~\eqref{eq.D.3} since the $\max_{1\leq j \leq k}\tilde \Pi^{-1}_j=O_P(1)$ and $\mid \tilde \pi_i-\pi_i\mid =o_p(1/N)$. Therefore the convergence of the second term of~\eqref{eq.vmvrs} has been proved.

Combining \eqref{term3}, \eqref{eq.D.3} and \eqref{term2}, we have proved that 
\begin{align*}
  \tilde\V_{N}^{\mvrs}=\V_{N}^{\str}+\op
\end{align*}
and Theorem~\ref{alg.thm} is proved.
\end{proof}

\section{Additional Numerical Results}\label{sim.res} 
  
In this section, we provide additional simulation experiments and real-world
data analysis results.

\subsection{Simulation results for IBOSS}\label{sim.iboss}
For logistic regression, IBOSS \citep{cheng2020information} is a prominent
alternative for reducing computational costs through deterministic, design-based
data selection. Although the proposed MVRS framework is primarily designed to
enhance random subsampling methods, comparing its performance against
deterministic approaches like IBOSS provides valuable insights into their
relative strengths.

We evaluate IBOSS under both correctly specified models and scenarios involving
outliers. Since IBOSS is a deterministic selection method, we generate a new
dataset for each simulation repetition, with MSEs calculated relative to the
true parameter as
$\operatorname{MSE}=R^{-1}\sum_{r=1}^R\| \htheta_{n,r}-\theta\|^2$. The baseline
simulation settings follow Cases 1--4 in Section 4.1. To assess
robustness, we introduce outliers via a label-flipping mechanism: for a given
proportion $\alpha \in \{0, 0.01, 0.05, 0.1\}$, the $\lceil \alpha N \rceil$
observations with the highest leverage scores have their responses $y$ flipped.

\begin{figure}[H]
    \centering
    \begin{subfigure}{0.45\linewidth}
        \centering
        \includegraphics[width=\linewidth]{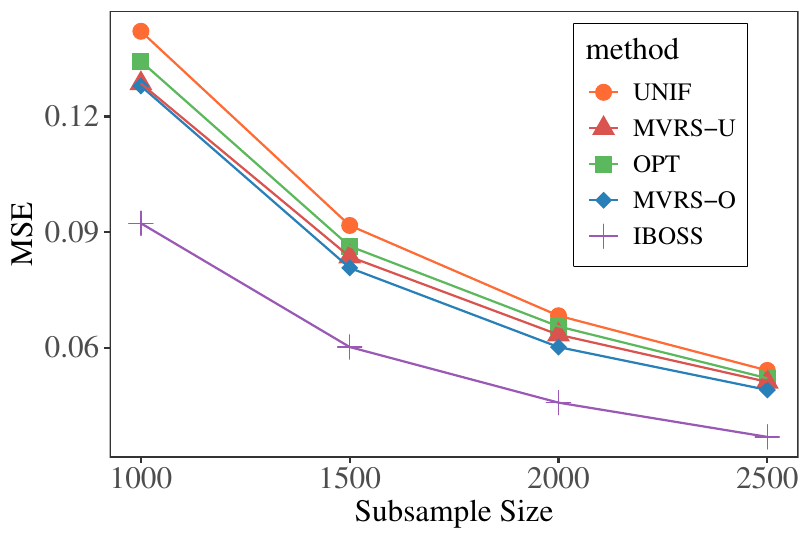}
        \caption{Case 1 (mzNormal)}
    \end{subfigure}
    \hspace{0.05\linewidth}
    \begin{subfigure}{0.45\linewidth}
        \centering
        \includegraphics[width=\linewidth]{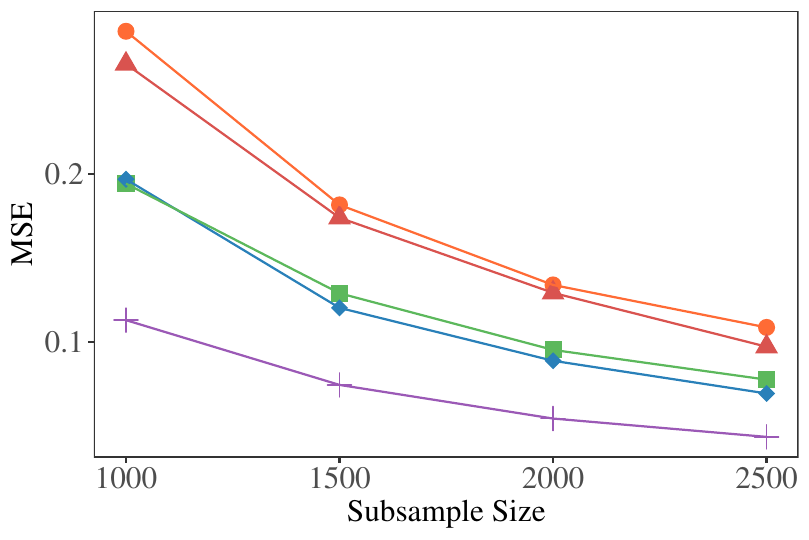}
        \caption{Case 2 (nzNormal)}
    \end{subfigure}
    \vfill
    \begin{subfigure}{0.45\linewidth}
        \centering
        \includegraphics[width=\linewidth]{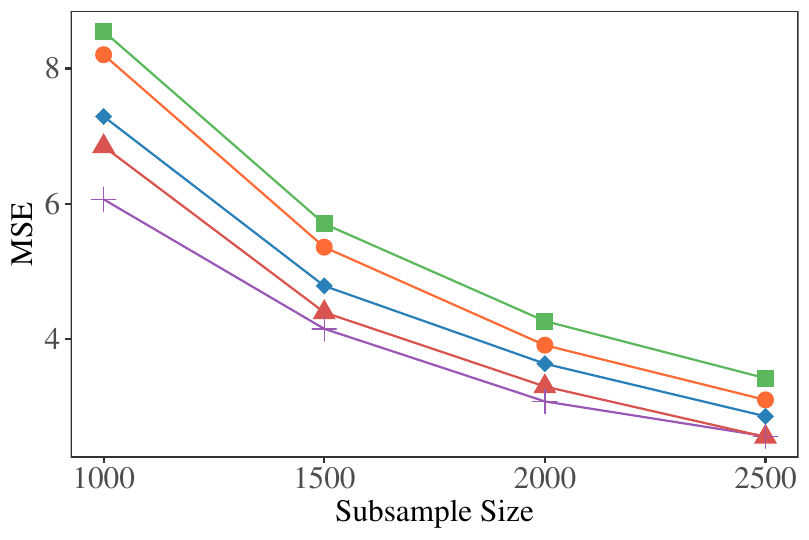}
        \caption{Case 3 (ueNormal)}
    \end{subfigure}
    \hspace{0.05\linewidth}
    \begin{subfigure}{0.45\linewidth}
        \centering
        \includegraphics[width=\linewidth]{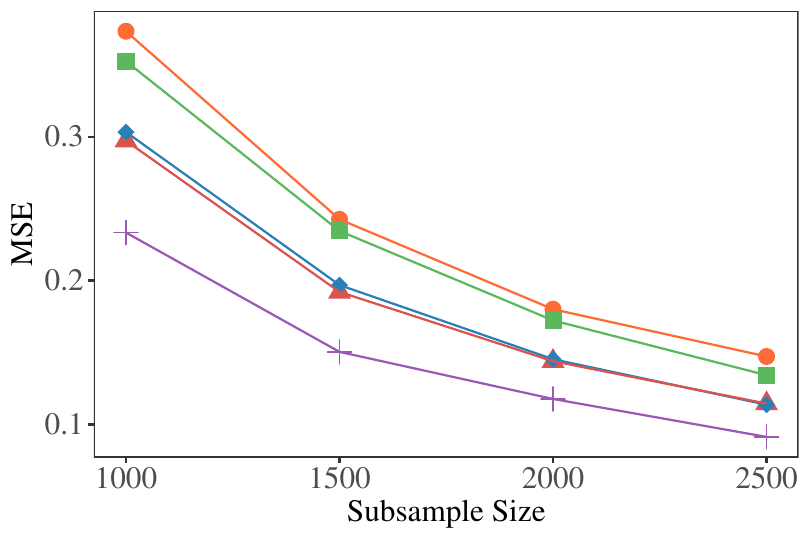}
        \caption{Case 4 (EXP)}
    \end{subfigure}
    \caption{MSEs for different subsample size $n$ in logistic regression.(IBOSS)}
    \label{fig:log.IBOSS}
\end{figure}

The results for the correctly specified model (Figure~\ref{fig:log.IBOSS})
indicate that IBOSS generally outperforms random subsampling-based
strategies. By deterministically selecting data points guided by the
D-optimality criterion, IBOSS does not bring in additional randomness and
achieves superior efficiency when model assumptions hold strictly. However, the
results under the outlier scenario (Figure~\ref{fig:log.IBOSS.out}) reveal a
critical trade-off. While the MSE for all methods increases with the outlier
proportion, random subsampling-based strategies exhibit greater
robustness. IBOSS, due to its reliance on extreme observations, is highly
sensitive to outliers.  The contaminated points in this scenario can skew the
deterministic selection, leading to substantial estimation bias, while the
stochasticity inherent in random subsampling methods provides a better safeguard
against such model misspecification.  Similar findings are also observed in the
real-world case study.

\begin{figure}[H]
    \centering
    \begin{subfigure}{0.45\linewidth}
        \centering
        \includegraphics[width=\linewidth]{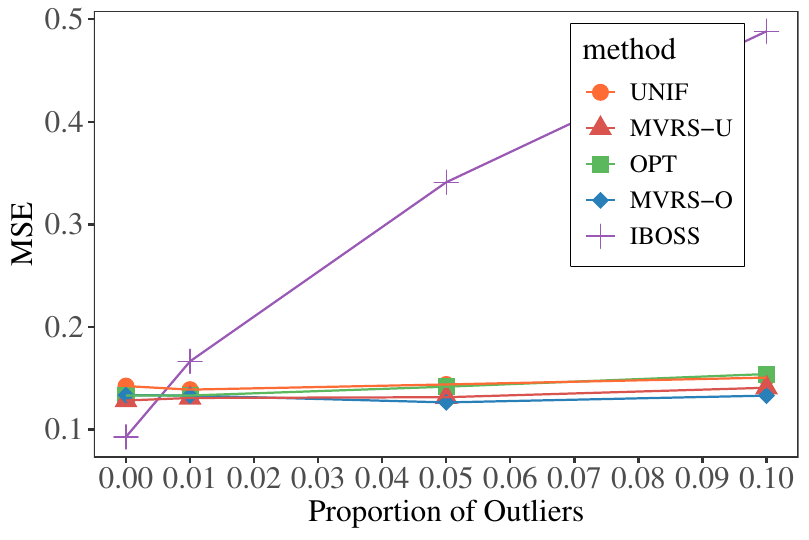}
        \caption{Case 1 (mzNormal)}
    \end{subfigure}
    \hspace{0.05\linewidth}
    \begin{subfigure}{0.45\linewidth}
        \centering
        \includegraphics[width=\linewidth]{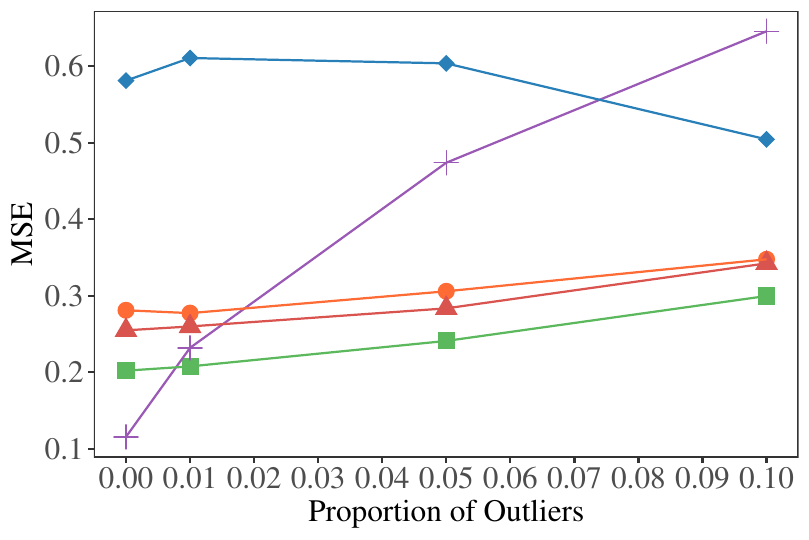}
        \caption{Case 2 (nzNormal)}
    \end{subfigure}
    \vfill
    \begin{subfigure}{0.45\linewidth}
        \centering
        \includegraphics[width=\linewidth]{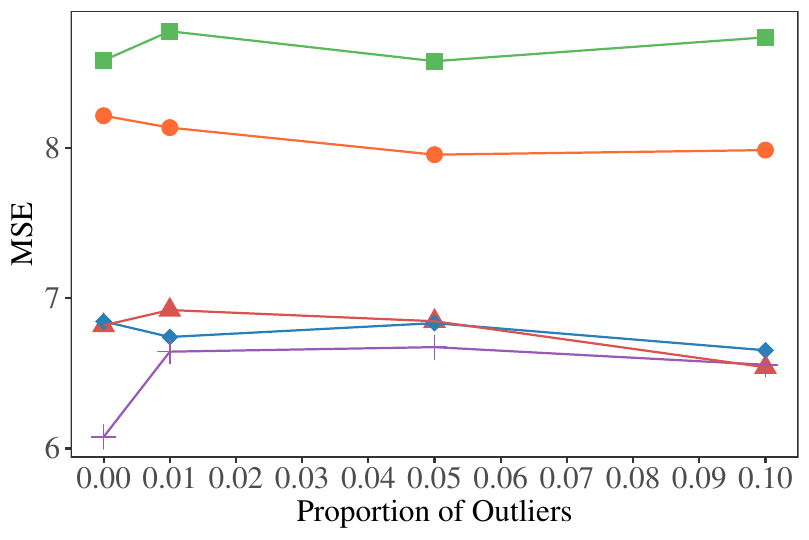}
        \caption{Case 3 (ueNormal)}
    \end{subfigure}
    \hspace{0.05\linewidth}
    \begin{subfigure}{0.45\linewidth}
        \centering
        \includegraphics[width=\linewidth]{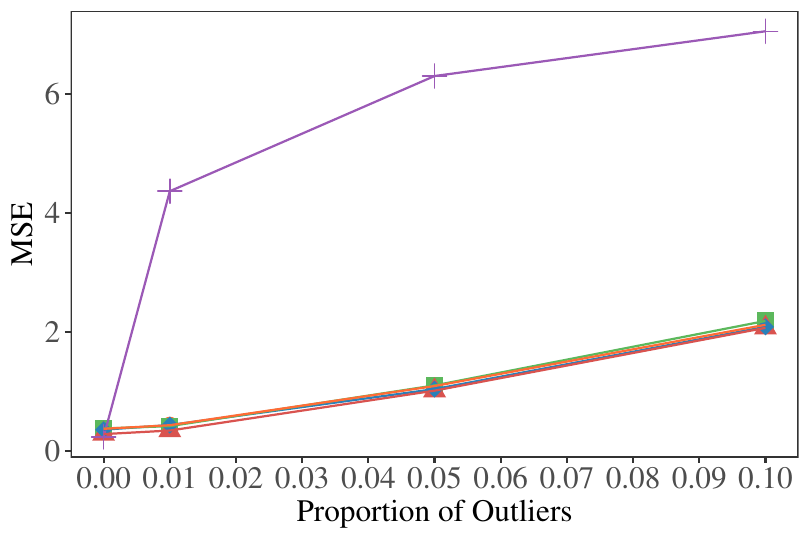}
        \caption{Case 4 (EXP)}
    \end{subfigure}
    \caption{MSEs for different proportion of outliers in logistic regression.(IBOSS)}
    \label{fig:log.IBOSS.out}
\end{figure}

While the MVRS framework effectively enhances random subsampling, integrating it
with deterministic methods like IBOSS presents an intriguing challenge. The
deterministic selection in IBOSS may conflict with the partitioning required for
stratification; for instance, \cite{wang2019divide} noted that dividing the full
dataset into non-overlapping blocks can lead to efficiency loss in
IBOSS. Consequently, a naive application of MVRS might cause an efficiency loss
for IBOSS. Furthermore, MVRS utilizes response information for
stratification. While this does not introduce sampling bias for random subsampling
methods due to the inverse probability weighting, it may bring in bias for
deterministic selection methods like IBOSS because they define estimators
through unweighted target functions. We leave the formal exploration of these
synergies as a direction for future research.

\subsection{Additional results for SUSY dataset}
In this section, we present additional results for the SUSY dataset.
Table~\ref{tab:cas_susy} presents the empirical SEs and Biases of individual
coefficients with $n=1000$ for the SUSY dataset. 

\begin{table}[H]
  
\centering
\caption{  Empirical SEs and Biases of individual coefficients with $n=1000$ for
  SUSY dataset. The SE ratio is calculated as the SE of MVRS divided
  by the SE of the corresponding baseline method.}
\label{tab:cas_susy}
\resizebox{1\textwidth}{!}{
\begin{tabular}{c|cc c cc c c|cc c cc c c}
    \hline
    \multirow{2}{*}{$\theta_i$}&
  \multicolumn{2}{c}{UNIF}&&\multicolumn{2}{c}{MVRS-U}&& &
  \multicolumn{2}{|c}{OPT}&& \multicolumn{2}{c}{MVRS-O}&&\\
    & SE & Bias && SE & Bias && SE Ratio & SE & Bias && SE & Bias &&SE Ratio\\
    \hline
1 & 0.018 & 0.040 && 0.017 & 0.034 && 0.982 & 0.011 & 0.007 && 0.010 & 0.006 && 0.928 \\
2 & 0.245 & 0.096 && 0.202 & 0.120 && 0.827 & 0.139 & 0.017 && 0.124 & -0.013 && 0.889 \\
3 & 0.008 & 0.002 && 0.009 & 0.001 && 1.103 & 0.009 & -0.005 && 0.009 & 0.000 && 0.983 \\
4 & 0.008 & 0.005 && 0.008 & 0.000 && 0.932 & 0.008 & -0.004 && 0.008 & -0.005 && 0.992 \\
5 & 0.063 & 0.018 && 0.065 & 0.049 && 1.043 & 0.040 & 0.012 && 0.037 & 0.005 && 0.914 \\
6 & 0.008 & 0.002 && 0.008 & -0.001 && 0.934 & 0.009 & 0.002 && 0.008 & -0.005 && 0.931 \\
7 & 0.008 & 0.001 && 0.007 & -0.005 && 0.918 & 0.008 & -0.001 && 0.008 & 0.000 && 1.062 \\
8 & 0.469 & 0.175 && 0.438 & 0.204 && 0.935 & 0.268 & 0.047 && 0.233 & 0.012 && 0.869 \\
9 & 0.007 & 0.000 && 0.007 & 0.000 && 0.929 & 0.007 & -0.001 && 0.007 & -0.002 && 0.948 \\
10 & 0.085 & -0.008 && 0.078 & -0.036 && 0.920 & 0.052 & -0.005 && 0.056 & -0.006 && 1.084 \\
11 & 0.125 & -0.005 && 0.088 & -0.024 && 0.705 & 0.070 & 0.007 && 0.058 & 0.011 && 0.821 \\
12 & 3.055 & 0.046 && 1.018 & 0.003 && 0.333 & 1.163 & -0.016 && 0.454 & 0.006 && 0.391 \\
13 & 0.478 & -0.072 && 0.440 & -0.130 && 0.921 & 0.248 & -0.051 && 0.213 & -0.021 && 0.857 \\
14 & 0.076 & -0.017 && 0.073 & -0.029 && 0.952 & 0.058 & -0.016 && 0.051 & -0.005 && 0.879 \\
15 & 0.131 & -0.012 && 0.116 & -0.023 && 0.886 & 0.075 & 0.019 && 0.073 & 0.010 && 0.969 \\
16 & 2.540 & -0.102 && 0.845 & -0.095 && 0.333 & 0.932 & 0.003 && 0.378 & 0.004 && 0.405 \\
17 & 0.228 & 0.022 && 0.205 & 0.064 && 0.900 & 0.136 & 0.002 && 0.111 & 0.005 && 0.817 \\
18 & 0.023 & -0.008 && 0.021 & 0.004 && 0.911 & 0.019 & 0.005 && 0.019 & 0.005 && 0.991 \\
19 & 0.030 & 0.001 && 0.026 & 0.005 && 0.869 & 0.024 & 0.008 && 0.023 & 0.012 && 0.946 \\
    \hline
\end{tabular}
}
\end{table}

\end{document}

%% file: commands.tex
\usepackage{graphicx, subcaption, amssymb, enumerate, natbib} 
 \usepackage[a4paper,scale=0.75]{geometry}
\usepackage{lineno} 
\usepackage{algorithm}
\usepackage{algorithmic}
\usepackage{xcolor,hyperref,float}
\usepackage{amsmath,amsthm}
\usepackage{multirow}
\usepackage{booktabs}

\usepackage{xr-hyper}
\makeatletter
\newcommand*{\addFileDependency}[1]{
  \typeout{(#1)}
  \@addtofilelist{#1}
  \IfFileExists{#1}{}{\typeout{No file #1.}}
}
\makeatother

\newtheorem{lemma}{Lemma}
\newtheorem{proposition}{Proposition}
\newtheorem{assumption}{Assumption}
\newtheorem{theorem}{Theorem}

\newcommand{\I}{\mathbf{I}}
\newcommand{\D}{\mathcal{D}}

\newcommand{\Q}{\mathbb{Q}}

\renewcommand{\u}{\mathbf{u}}

\newcommand{\V}{\mathbf{V}}

\newcommand{\htheta}{{\hat{{\theta}}}}

\newcommand{\onen}{\frac{1}{n}}
\newcommand{\oneN}{\frac{1}{N}}

\newcommand{\op}{o_P(1)}

\newcommand{\Op}{O_p(1)}

\newcommand{\ud}{\mathrm{d}}
\newcommand{\DN}{\mathcal{D}_N}
\newcommand{\sumn}{\sum_{i=1}^{n}}
\newcommand{\sumN}{\sum_{i=1}^{N}}

\newcommand{\tp}{^{\rm T}}

\renewcommand{\Pr}{\mathbb{P}}
\newcommand{\Exp}{\mathbb{E}}
\newcommand{\Nor}{\mathbb{N}}
\newcommand{\Var}{\mathbb{V}}

\newcommand{\str}{{\mathrm{str}}}
\newcommand{\mvrs}{{\mathrm{mvrs}}}
\newcommand{\sub}{{\mathrm{sub}}}
\newcommand{\mI}{{\mathbb{I}}}

\newcommand{\keywords}[1]{\par\textbf{Keywords: } #1}